\documentclass[journal]{IEEEtran}
\IEEEoverridecommandlockouts
\usepackage[english]{babel}
\usepackage{graphicx}
\usepackage{epstopdf}
\usepackage{hyperref}
\usepackage[multiple]{footmisc}
\usepackage{amsmath,amsfonts, amssymb}
\usepackage{amsthm}

\newtheorem{theorem}{Theorem}
\newtheorem{corollary}{Corollary}
\newtheorem{lemma}{Lemma}
\newtheorem{remark}{Remark}
\newtheorem{definition}{Definition}
\newtheorem{example}{Example}
\usepackage{epsfig} % for postscript graphics files
\usepackage{psfrag}

\usepackage{times}
\usepackage[tight,footnotesize]{subfigure}
\usepackage{bibspacing}

\usepackage{color}

\usepackage{verbatim}
\usepackage{babel}
\usepackage{listings}
\usepackage{comment}
\usepackage{multirow}
\usepackage[ruled]{algorithm2e}
\usepackage[tight]{subfigure}
\usepackage{mathtools}
\DeclarePairedDelimiter\ceil{\lceil}{\rceil}

\allowdisplaybreaks
\begin{document}
\title{\LARGE \bf Taxi Dispatch with Real-Time Sensing Data in Metropolitan Areas: A Receding Horizon Control Approach}

\author{Fei~Miao,
            ~\IEEEmembership{Student Member,~IEEE,}
              Shuo~Han,
              ~\IEEEmembership{Member,~IEEE,}
             Shan~Lin,
             John A.~Stankovic,
            ~\IEEEmembership{Fellow,~IEEE and ACM,}
              Desheng~Zhang, 
             %  ~\IEEEmembership{Member,~IEEE,}  
             Sirajum~Munir,
             Hua~Huang,
             Tian~He,
             and~George~J.~Pappas
            ~\IEEEmembership{Fellow,~IEEE}
\thanks{This work was supported by NSF grant numbers CNS-1239483, CNS-1239108, CNS-1239226, and CPS-1239152 with project title: CPS: Synergy: Collaborative Research: Multiple-Level Predictive Control of Mobile Cyber Physical Systems with Correlated Context. The preliminary conference version of this work can be found in~\cite{taxi_iccps15}.}             
\thanks{F. Miao, S. Han and G. J. Pappas are with Department of Electrical and Systems Engineering, University of Pennsylvania, Philadelphia PA, 19104, USA 19014.  Email: \{miaofei, hanshuo, pappasg\}@seas.upenn.edu.}
\thanks{S.~Lin and H.~Huang are with Department of Electrical and Computer Engineering, Stony Brook University, Long Island, NY 11794, USA. Email:\{shan.x.lin, hua. huang\}@stonybrook.edu\}.}
\thanks{J.A.~Stankovic and S.~Munir are with Department of Computer Science, University of Virginia, Charlottesville, VA 22904, USA. Email: stankovic@cs.virginia.edu, sm7hr@virginia.edu.}
\thanks{D.~Zhang and T.~He are with Department of Computer Science and Engineering, University of Minnesota, Minneapolis, MN 55455, USA. Email:\{tianhe, zhang\}@cs.umn.edu}
}
%%%%%%%%%%%
%%%%%%%%%%%%


\maketitle

\label{abstract}
\begin{abstract}
Traditional taxi systems in metropolitan areas often suffer from inefficiencies due to uncoordinated
actions as system capacity and customer demand change. With the pervasive deployment of networked sensors in modern vehicles, large amounts of information regarding customer demand and system status can be collected in real time. This information provides opportunities to perform various types of control and coordination for large-scale
intelligent transportation systems. In this paper, we present a receding horizon control (RHC)
framework to dispatch taxis, which incorporates highly spatiotemporally correlated demand/supply models and real-time GPS location and occupancy information. The objectives include matching spatiotemporal ratio between demand and supply for service quality with minimum current and anticipated future taxi idle driving distance. Extensive trace-driven analysis with a data set containing taxi operational records in San Francisco shows that our solution
reduces the average total idle distance by $52\%$, and reduces the supply demand ratio error across the city during one experimental time slot by $45\%$. 
Moreover, our RHC framework is compatible with a wide variety of predictive models and optimization problem formulations. This compatibility property allows us to solve robust optimization problems with corresponding demand uncertainty models that provide disruptive event information. 
\iffalse
Our framework dispatches taxi towards current and future passengers
based on a novel combination of
historical spatiotemporal passenger demand (thereby incorporating information such as
rush hours, bottlenecks, train arrivals, etc.) and current information. 
\fi
\end{abstract}

\renewcommand{\abstractname}{Note to Practitioners}
\begin{abstract}
With the development of mobile sensor and data processing technology, the competition between traditional \lq \lq hailed on street'' taxi service and \lq \lq on demand'' taxi service has emerged in the US and elsewhere. In addition, large amounts of data sets for taxi operational records provide potential demand information that is valuable for better taxi dispatch systems. Existing taxi dispatch approaches are usually greedy algorithms focus on reducing customer waiting time instead of total idle driving distance of taxis. Our research is motivated by the increasing need for more efficient, real-time taxi dispatch methods that utilize both historical records and real-time sensing information to match the dynamic customer demand. This paper suggests a new receding horizon control (RHC) framework aiming to utilize the predicted demand information when making taxi dispatch decisions, so that passengers at different areas of a city are fairly served and the total idle distance of vacant taxis are reduced. We formulate a multi-objective optimization problem based on the dispatch requirements and practical constraints. The dispatch center updates GPS and occupancy status information of each taxi periodically and solves the computationally tractable optimization problem at each iteration step of the RHC framework. Experiments for a data set of taxi operational records in San Francisco show that the RHC framework in our work can redistribute taxi supply across the whole city while reducing total idle driving distance of vacant taxis. In future research, we plan to design control algorithms for various types of demand model and experiment on data sets with a larger scale.
\end{abstract}
%\category{H.4}{Information System Application}{Miscellaneous}
%\category{I.2.8}{Problem Solving, Control Methods, and Search}{Control Application, Transportation}
%\terms{Algorithms, Design, Experimentation, Performance}
%\keywords{Intelligent Transportation System, Real-Time Taxi Dispatch, Receding Horizon Control, Mobility Pattern}

\section{Introduction}
\label{intro}
%\textbf{Related work}:

More and more transportation systems are equipped with various sensors and wireless radios to enable novel mobile cyber-physical systems, such as intelligent highways, traffic light control, supply chain management, and autonomous fleets. The embedded sensing and control technologies in these systems significantly improve their safety and efficiency over traditional systems. In this paper, we focus on modern taxi networks, where real-time occupancy status and the Global Positioning System (GPS) location of each taxi are sensed and collected to the dispatch center. Previous research has shown that such data contains rich information about passenger and taxi mobility patterns~\cite{Dmodel, sf_data, prevacant}. Moreover, recent studies have shown that the passenger demand information can be extracted and used to reduce passengers' waiting time, taxi cruising time, or future supply rebalancing cost to serve requests~\cite{dataset_strategy, Powell2011, mod}. 

% for higher service quality, and lower idle driving distance.

Efficient coordination of taxi networks at a large scale is a challenging task. Traditional taxi networks in metropolitan areas largely rely on taxi drivers' experience to look for passengers on streets to maximize individual profit. However, such self-interested, uncoordinated behaviors of drivers usually result in spatiotemporal mismatch between taxi supply and passenger demand. In large taxi companies that provide dispatch services, greedy algorithms based on locality are widely employed, such as finding the nearest vacant taxi to pick up a passenger~\cite{rt_gps}, or first-come, first-served. Though these algorithms are easy to implement and manage, they prioritize immediate customer satisfaction over global resource utilization and service fairness, as the cost of rebalancing the entire taxi supply network for future demand is not considered. 
%global performance is sacrificed, and these solutions do not consider possible future costs to rebalance the taxi supply throughout the entire network. 

\iffalse
Both global service fairness -- balancing vehicle availabilities throughout the network and anticipated future rebalancing costs are analyzed in a queuing-theoretical model~\cite{mod}, but taxis' idle driving cost is not considered. Both costs of idle cruising and missing tasks are included to assign trucks in the temporal perspective~\cite{rtpick_Patrick}, but real-time location information is not involved in the assignment.
Other control methods such as coverage control and coordination~\cite{ccms_Martinez} for groups of vehicles are also developed for allocating supply to match demand. %, but fail to give a clear framework . 
However, none of these approaches incorporate both system model learned from historical data and real-time GPS information into a control framework, that considers both current and future anticipated idle driving cost and global service fairness at the same time. In this work, we consider this new design problem. 
\fi
 
Our goal is to utilize real-time information to optimize taxi network dispatch for anticipated future idle driving cost and global geographical service fairness, while fulfilling current, local passenger demand. To accomplish such a goal, we incorporate both system models learned from historical data and real-time taxi data into a taxi network control framework. To the best of our knowledge, this is the first work to consider this problem. The preliminary version of this work can be found in~\cite{taxi_iccps15}, and more details about problem formulation, algorithm design, and numerical evaluations are included in this manuscript. 

%Learning a mobile network model~\cite{Dmodel} provides a basis for control strategies. 
%missing an opportunity for enhanced and coordinated system performance.
In this paper, we design a computationally efficient moving time horizon framework for taxi dispatch with large-scale real-time information of the taxi network. Our dispatch solutions in this framework consider future costs of balancing the supply demand ratio under realistic constraints. We take a receding horizon control (RHC) approach to dynamically control taxis in large-scale networks. Future demand is predicted based on either historical taxi data sets~\cite{bootstrap} or streaming data~\cite{Dmodel}. The real-time GPS and occupancy information of taxis is also collected to update supply and demand information for future estimation. This design iteratively regulates the mobility of idle taxis for high performance, demonstrating the capacity of large-scale smart transportation management. 
\iffalse
Each iteration of the MPC method runs a multi-objective optimization problem, that considers both idle driving cost and reaching a balanced supply/demand ratio in several consecutive time slots. \fi

%where the latter dispatches vacant taxis drive idly towards passengers. 
\iffalse as feedback to measure control performance and adjust dispatch decisions. \fi

\iffalse
Previous multi-agent service algorithm, like coverage control usually depends on the service demand received by the agents or control center, that each robot/taxi in the network is responsible  for requests in its coverage region according to a distributed control algorithm.  Our work however, considers costs to satisfy both real time taxi position information, service requests, and possible future requests when dispatching available taxis.
\fi

The contributions of this work are as follows, 
\begin{itemize}
%\vspace{-6pt}
\item To the best of our knowledge, we are the first to design an RHC framework for large-scale taxi dispatching. We consider both current and future demand, saving costs under constraints by involving expected future idle driving distance for re-balancing supply.
\item The framework incorporates large-scale data in real-time control. Sensing data is used to build predictive passenger demand, taxi mobility models, and serve as real-time feedback for RHC. 
\item Extensive trace driven analysis based on a San Francisco taxi data set shows that our approach reduces average total taxi network idle distance by $52\%$ as in Figure~\ref{compare_gps}, and the error between local and global supply demand ratio by $45\%$ as in Figure~\ref{compare_ratio}, compared to the actual historical taxi system performance. 
\item Spatiotemporal context information such as disruptive passenger demand is formulated as uncertainty sets of parameters into a robust dispatch problem. This allows the RHC framework to provide more robust control solutions under uncertain contexts as shown in Figure~\ref{fig:robust}. The error between local and global supply demand ratio is reduced by $25\%$ compared with the error of solutions without considering demand uncertainties. 
%\item 
%\item Spatiotemporal context information such as disruptive passenger demand is incorporated into our control framework. This allows our control solutions to be more robust and accurate to such disturbances under uncertain contexts as shown in Figure~\ref{fig:robust}.
%\vspace{-6pt}
\end{itemize}

\subsection{State-of-the-Art}
\label{related_work}
There are three categories of research topics related to our work: taxi dispatch systems, transportation system modeling, and multi-agent coordination and control. 

A number of recent works study approaches of taxi dispatching services or allocating transportation resources in modern cities. Zhang and Pavone~\cite{mod} designed an optimal rebalancing method for autonomous vehicles, which considers both global service fairness and future costs, but they didn't take idle driving distance and real-time GPS information into consideration. Truck schedule methods to reduce costs of idle cruising and missing tasks are designed in the temporal perspective in work~\cite{rtpick_Patrick}, but the real-time location information is not utilized in the algorithm. Seow et.al\ focus on minimizing total customer waiting time by concurrently dispatching multiple taxis and allowing taxis to exchange their booking assignments~\cite{multiagent_taxi}. A shortest time path taxi dispatch system based on real-time traffic conditions is proposed by Lee et.al~\cite{Lee_review}. In~\cite{recommend,Huang:2012,Powell2011}, authors aim to maximize drivers' profits by providing routing recommendations. These works give valuable results, but they only consider the current passenger requests and available taxis. Our design uses receding horizon control to consider both current and predicted future requests. %along with the pervasive deployment of sensing technologies 

Various mobility and vehicular network modeling techniques have been proposed for transportation systems~\cite{traffic_model, time_Work}. Researchers have developed methods to predict travel time~\cite{time_predict, time_bayen} and traveling speed~\cite{Jaillet}, and to characterize taxi performance features~\cite{dataset_strategy}. A network model is used to describe the demand and supply equilibrium in a regulated market is investigated~\cite{Yang_review}. These works provide insights to transportation system properties and suggest potential enhancement on transportation system performance. Our design takes a step further to develop dispatch methods based on available predictive data analysis.

There is a large number of works on mobility coordination and control. Different from taxi services, these works usually focus on region partition and coverage control so that coordinated agents can perform tasks in their specified regions~\cite{ccms_Martinez, survey_koutsoukos, mobile_network}. Aircraft dispatch system and air traffic management in the presence of uncertainties have been addressed~\cite{ddd_air, air_traffic}, while the task models and design objectives are different from taxi dispatching problem. Also, receding horizon control (RHC) has been widely applied for process control, task scheduling, and multi-agent transportation networks~\cite{mpc_transport, traffic_mpc, rhc_multiagent}.  
%\cite{mpc_johansson, mpc_cruise, mpc_transport, traffic_mpc}. 
These works provide solid results for related mobility scheduling and control problems. However, none of these works incorporates both the real-time sensing data and historical mobility patterns into a receding horizon control design, leveraging the taxi supply based on the spatiotemporal dynamics of passenger demand. 

The rest of the paper is organized as follows. The background of taxi monitoring system and control problems are introduced in Section~\ref{sec:prob_descrip}. The taxi dispatch problem is formally formulated in Section~\ref{sec:prob_form}, followed by the RHC framework design in Section~\ref{sec:algorithm}.
A case study with a real taxi data set from San Francisco to evaluation the RHC framework is shown in Section~\ref{sec:simulation}.
%with different parameters, and comparison of dispatched with non-dispatched taxi network measurements are shown in Section~\ref{sec:simulation}.
Concluding remarks are made in Section~\ref{sec:conclusion}.

\section{Taxi Dispatch Problem: Motivation and System}
\label{sec:prob_descrip}
%\subsection{Motivation}
%There exists a gap between previous coverage control and real world data, that
\begin{figure}[b!]
\vspace{-15pt}
\centering
\includegraphics [width=0.39\textwidth]{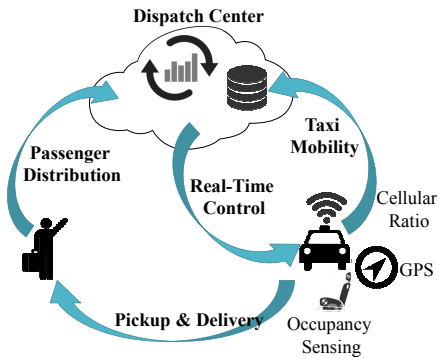}
\vspace{-8pt}
\caption{A prototype of the taxi dispatch system}
%\vspace{-8pt}
\label{sys_structure}
%\vspace{-10pt}
\end{figure}
%\subsection{System Structure}
Taxi networks provide a primary transportation service in modern cities. Most street taxis respond to passengers' requests on their paths when passengers hail taxis on streets. %, and take passengers to their specified destinations. Therefore, existing taxi networks rely on drivers to \mbox{arbitrarily} pick up passengers on streets. 
This service model has successfully served up to 25\% public passengers in metropolitan areas, such as San Francisco and New York~\cite{sf_survey,taxi_tomorrow}. However, passenger's waiting time varies at different regions of one city and taxi service is not satisfying. In the recent years, "on demand"  transportation service providers like Uber and Lyft aim to connect a passenger directly with a driver to minimize passenger's waiting time. This service model is still uncoordinated, since drivers may have to drive idly without receiving any requests, and randomly traverse to some streets in hoping to receive a request nearby based on experience.

Our goal in this work is a centralized dispatch framework to coordinate service behavior of large-scale taxi Cyber-Physical system. The development of sensing, data storage and processing technologies provide both opportunities and challenges to improve existing taxi service in metropolitan areas. Figure~\ref{sys_structure} shows a typical monitoring infrastructure, which consists of a dispatch center and a large number of geographically distributed sensing and communication components in each taxi. The sensing components include a GPS unit and a trip recorder, which provides real-time geographical coordinates and occupancy status of every taxi to the dispatch center via cellular radio. The dispatch center collects and stores data. Then, the monitoring center runs the dispatch algorithm to calculate a dispatch solution and sends decisions to taxi drivers via cellular radio. Drivers are notified over the speaker or on a special display. 
%Besides location, time, and occupancy status, this taxi monitoring system also provides information on environmental conditions, road conditions, and \textit{in situ} road pictures.

Given both historical data and real-time taxi monitoring information described above, we are capable to learn spatiotemporal characteristics of passenger demand and taxi mobility patterns. This paper focuses on the dispatch approach with the model learned based on either historical data or streaming data. One design requirement is balancing spatiotemporal taxi supply across the whole city from the perspective of system performance. %Similar requirement of balancing server node utilization rate is analyzed in~\cite{mod}. 
 It is worth noting that heading to the allocated position is part of idle driving distance for a vacant taxi. Hence, there exists trade-off between the objective  of matching supply and demand and reducing total idle driving distance. We aim at a scalable control framework that directs vacant taxis towards demand, while balancing between minimum current and anticipated future idle driving distances.

\section{Taxi Dispatch Problem Formulation}
\label{sec:prob_form}
%According to the design requirements, the taxi dispatch 

Informally, the goal of our taxi dispatch system is to schedule vacant taxis towards predicted passengers both spatially and temporally with minimum total idle mileage. We use supply demand ratio of different regions within a time period as a measure of service quality, since sending more taxis for more requests is a natural system-level requirement to make customers at different locations equally served. Similar service metric of service node utilization rate has been applied in resource allocation problems, and autonomous driving car mobility control approach~\cite{mod}.
%Based on the spatiotemporal patterns of passenger demands in the city, the dispatch center dynamically allocates vacant taxis to different regions in order to match the passenger demands.  %(cite some other similar resource allocation work)   

%To predict the passenger demand accurately, 
 The dispatch center receives real-time sensing streaming data including each taxi's GPS location and occupancy status with a time stamp periodically. The real-time data stream is then processed at the dispatch center to predict the spatiotemporal patterns of passenger demand. Based on the prediction, the dispatch center calculates a dispatch solution in real-time, and sends decisions to vacant taxis with dispatched  regions to go in order to match predicted passenger demands.
 %Learning techniques to predict spatiotemporal patterns of passenger demand, either purely from history data, or streaming data is equipped in the dispatch center. 

%and dispatching suggestion is sent in discrete time, so 
\iffalse
Because it is not computational efficient to estimate a demand model directly based on continuous time and two dimensional geometric space for a large-scale data set, we assume that both time and position space are discretized for a spatiotemporal demand model.  
\fi

%both current and anticipated future demand.

Besides balancing supply and demand, another consideration in taxi dispatch is minimizing the total idle cruising distance of all taxis. A dispatch algorithm that introduces large idle distance in the future after serving current demands can decrease total profits of the taxi network in the long run. Since it is difficult to perfectly predict the future of a large-scale taxi service system in practice, we use a heuristic estimation of idle driving distance to describe anticipated future cost associated with meeting customer requests.  
Considering control objectives and computational efficiency, we choose a receding horizon control approach.   
We assume that the optimization time horizon is $T$, indexed by $k=1,\dots, T$, given demand prediction during time $[1,T]$.  
\subsection*{Notation}
In this paper, we denote $\mathbf{1}_N$ as a length $N$ column vector of all $1$s, and $\mathbf{1}_N^T$ is the transpose of the vector.  Superscripts of variables as in $X^k, X^{k+1}$ denote discrete time. We denote the $j$-th column of matrix $X^k$ as $X_{\cdot j}^k$. 

\subsection{Supply and demand in taxi dispatch}
\label{sd_model}
\begin{table*}
\centering
%\caption{Parameters of the hybrid stochastic game between the system and the attacker}
\begin{tabular}{|c|c|}
  \hline
  Parameters&Description \\ \hline
  $N$ & the total number of vacant taxis\\ \hline
  $n$ & the number of  regions\\ \hline       
   $r^k \in \mathbb{R}^{1 \times n}$& the total number of predicted requests to be served by current vacant taxis at each region   \\ \hline
   $C^k \in [0,1] ^{n \times n}$ & matrix that describes taxi mobility patterns during one time slot\\ \hline %$C^k_{ij}=$ the probability that given a taxi starts at region $i$, it will reach region $j$ by the end of time $k$\\ \hline
   $P^0 \in \mathbb{R}^{N\times 2}$ & the initial positions of vacant taxis provided by GPS data\\ \hline
   $W_i \in \mathbb{R}^{n\times 2}$& preferred positions of the $i$-th taxi at $n$ regions\\\hline
   $\alpha \in \mathbb{R}^{N} $ & the upper bound of distance each taxi can drive for balancing the supply\\ \hline
   $\beta \in \mathbb{R}_+$ & the weight factor of the objective function \\ \hline  
   $R^k \in \mathbb{R}_+$ & total number of predicted requests in the city\\ \hline                                 
       Variables&Description \\ \hline    
      $X^k \in \{0,1\}^{N \times n}$ &  the dispatch order matrix that represents the region each vacant taxi should go \\ \hline     
    $P^k \in [0, 1]^{N \times n}$ & predicted positions of dispatched taxis at the end of time slot $k$ \\ \hline
    $d_i^k \in \mathbb{R}_+$ & lower bound of idle driving distance of the $i$-th taxi for reaching the dispatched location \\ \hline
%     $u^k \in \mathbb{R}^{N \times 1}$ & a vector that aggregates $u_i^k, i=1,2,\dots, N$\\ \hline             
\end{tabular}
%\vspace{-8pt}
\caption{Parameters and variables of the RHC problem~\eqref{opt2}.}
\label{opt_parameter}
%\vspace{-8pt}
\end{table*} 
%With a large amount of historical data on taxi GPS and occupancy status, we extract basic spatiotemporal characteristics of passenger demand by existing techniques like bootstrap~\cite{bootstrap} and online learning method~\cite{Dmodel}. 
We assume that the entire area of a city is divided into $n$ regions such as administrative sub-districts. We also assume that within a time slot $k$, the total number of passenger requests at the $j$-th region is denoted by $r^k_j$. We also use $r^k \triangleq [r^k_1,\dots, r^k_n] \in \mathbb{R}^{1 \times n}$ to denote the vector of all requests.  These are the demands we want to meet during time $k=1,\dots, T$ with minimal idle driving cost. Then the total number of predicted requests in the entire city is denoted by 
\\\centerline{$
%\begin{align*}
R^k=\sum\limits_{j=1}^{n}r^k_j.
%\end{align*}
$}
\iffalse 
The passenger drop-off location is related to the pick-up location in general. We assume that taxi mobility pattern during time slot $k$ is described by a region transition matrix $C^k \in \mathbb{R}^{n\times n}$ satisfying $\sum_{j=1}^{n} C_{ij}=1$,
where $C^k_{ij}$ is the probability that a taxi drops off a passenger at region $j$ near the end of time $k$ when it starts from region $i$.
From the queueing-theoretical perspective such probability transition matrix approximately describes passenger's mobility~\cite{mod}. Both $r^k$ and $C^k$ describe the demand model.
\fi

We assume that there are total $N$ vacant taxis in the entire city that can be dispatched according to the real-time occupancy status of all taxis. The initial supply information consists of real-time GPS position of all available taxis, denoted by $P^0 \in \mathbb{R}^{N\times 2}$, whose $i$-th row $P_i^0 \in \mathbb{R}^{1\times 2}$ corresponds to the position of the $i$-th vacant taxi.  While the dispatch algorithm does not make decisions for occupied taxis,  information of occupied taxis affects the predicted total demand to be served by vacant taxis, and the interaction between the information of occupied taxis and our dispatch framework will be discussed in section~\ref{sec:algorithm}.

The basic idea of the dispatch problem is illustrated in Figure~\ref{optresult}. Specifically, each region has a predicted number of requests that need to be served by vacant taxis, as well as locations of all vacant taxis with IDs given by real-time sensing information. We would like to find a dispatch solution that balances the supply demand ratio, while satisfying practical constraints and not introducing large current and anticipated future idle driving distance. Once dispatch decisions are sent to vacant taxis, the dispatch center will wait for future computing a new decision problem until updating sensing information in the next period.

\begin{figure}[!t]
%\begin{figure}[b!]
\vspace{-8pt}
\centering
\subfigure[A dispatch solution -- taxi $2$ goes to region $4$, and taxi $4$ goes to region $4$. ]
{\includegraphics [width=0.22\textwidth]{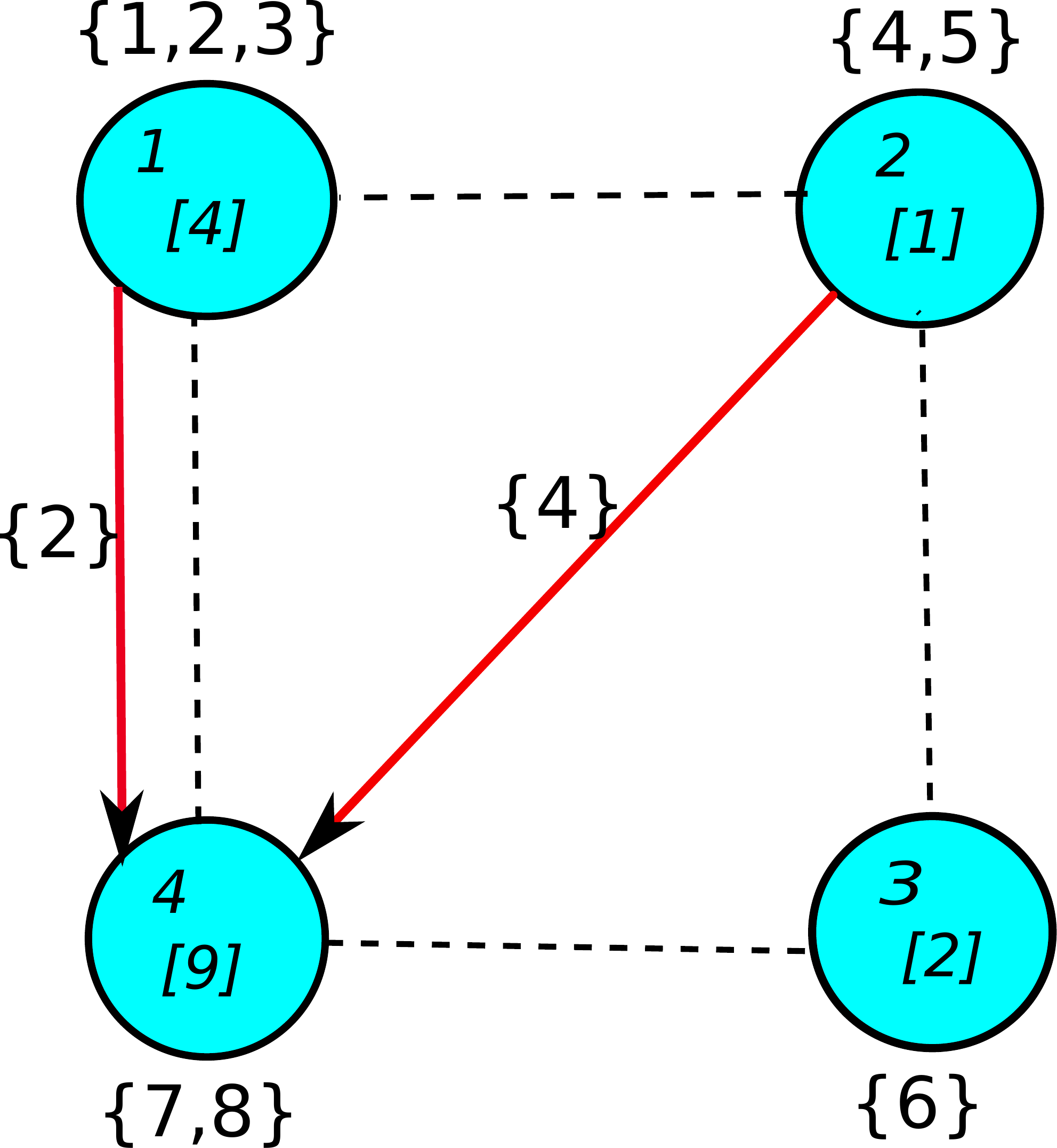}
\label{t2}
}
\subfigure[A dispatch solution -- taxi $2$ goes to region $4$, taxi $4$ goes to region $3$, and taxi $6$ goes to region $4$. ]
{\includegraphics [width=0.23\textwidth]{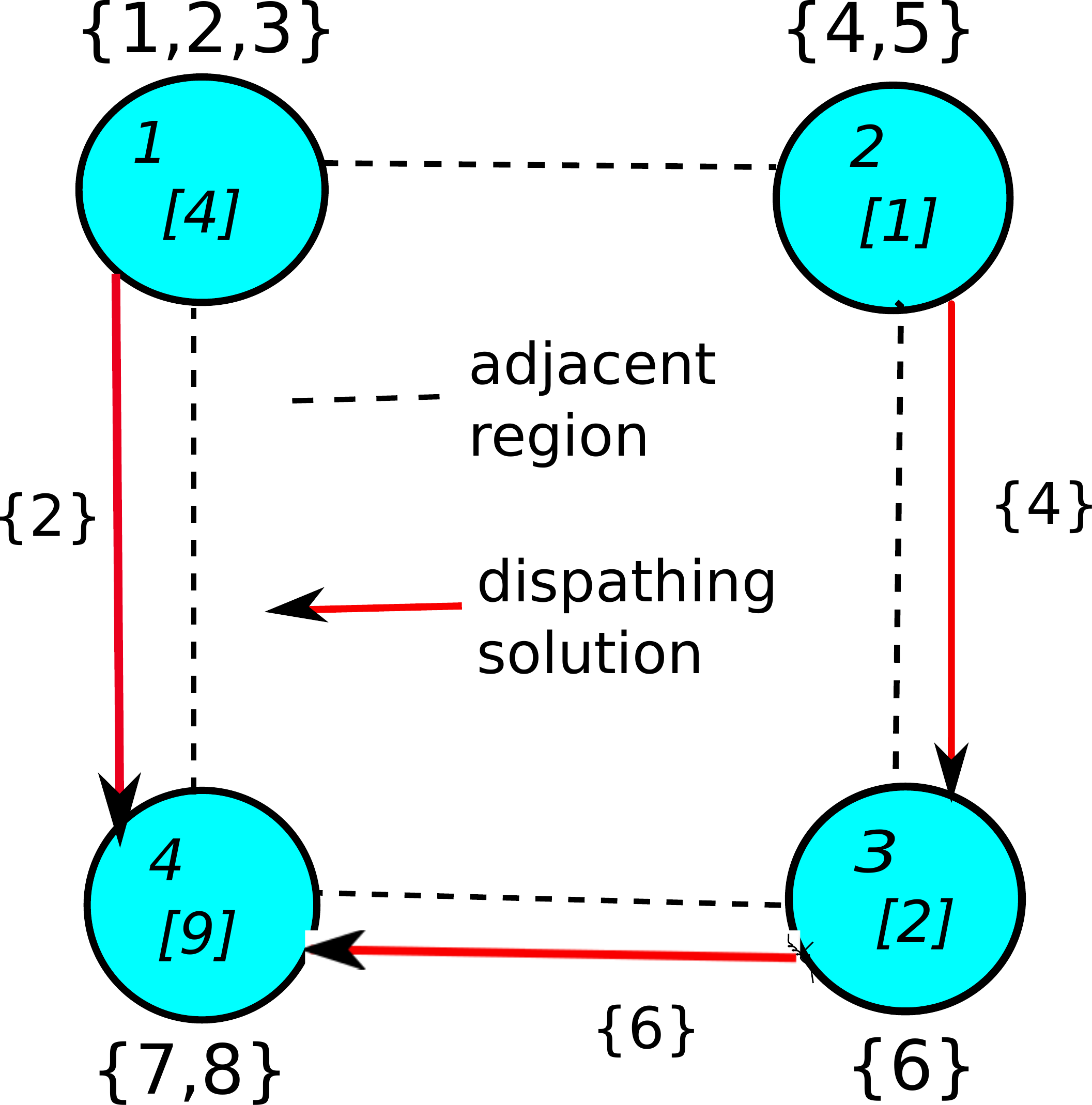}
\label{t3}
}
\vspace{-5pt}
\caption{Unbalanced supply and demand at different regions before dispatching and possible dispatch solutions. A circle represents a region, with a number of predicted requests ($[\cdot]$ inside the circle) and vacant taxis ($\{$ taxi IDs $\}$ outside the circle) before dispatching. A black dash edge means adjacent regions. A red edge with a taxi ID means sending the corresponding vacant taxi to the pointed region according to the predicted demand.} 
\label{optresult}
\vspace{-8pt}
\end{figure}

%\subsection{Variables}
\subsection{Optimal dispatch under operational constraints}
%\subsection{Variables, constraints and cost function}
The decision we want to make is the region each vacant taxi should go.  With the above initial information about supply and predicted demand, we define a  binary matrix $X^k \in\{0,1\}^{N \times n}$ as the dispatch order matrix, where $X^k_{ij}=1$ if and only if the $i$-th taxi is sent to the $j$-th region during time $k$. Then
\\\centerline{$
%\begin{align*} 
X^k \mathbf{1}_n=\mathbf{1}_N,\quad k=1,\dots, T
%\end{align*}
$}
must be satisfied, since every taxi should be dispatched to one region during time $k$.

\subsubsection{Two objectives}
%\textbf{Penalty for violating service fairness}:
One design requirement is to fairly serve the customers at different regions of the city --- vacant taxis should be allocated to each region according to predicted demand across the entire city during each time slot. To measure how supply matches demand at different regions, we use the metric---supply demand ratio. For region $j$, its supply demand ratio for time slot $k$ is defined as the total number of vacant taxis decided by the total number of customer requests during time slot $k$.
When the supply demand ratio of every region equals to that of the whole city, we have the following equations for $j=1, \dots, n$,  $k=1,\dots,T$,%or vacant taxis reach an equilibrium distribution, 
\begin{align}
\frac{\mathbf{1}_N^T X_{\cdot j}^k}{r_j^k}=\frac{N}{ R^k}, \iff \frac{\mathbf{1}_N^T X_{\cdot j}^k}{N}=\frac{r_j^k}{ R^k},
\label{xr}
\end{align}
For convenience, we rewrite equation~\eqref{xr} as the following equation about two row vectors 
\begin{align}
\frac{1}{N} \mathbf{1}^T_N X^k=\frac{1}{R^k}r^k, \quad k=1,\cdots,T.
\label{equal}
\end{align} 
\iffalse of vacant taxi distribution for time $k$--the percentage of available taxis in every region equals to the percentage of requests at every region. It means the supply/demand ratio of each region equals to that of the whole city, and reaches the steady state. \fi
%i.e., %as shown in Figure~\ref{equal}, i.e., 
%But for an integer matrix $X^k$, %number of vacant taxis and requests at each region, 
However, equation~\eqref{equal} can be too strict if used as a constraint, and there may be no feasible solutions satisfying~\eqref{equal}. This is because decision variables $X^k, k=1,\dots, T$ are integer matrices, and taxis' driving speed is limited that they may not be able to serve the requests from any arbitrary region during time slot $k$. Instead, we convert the constraint~\eqref{equal} into a soft constraint by introducing a supply-demand mismatch penalty function $J_E$ for the requirement that the supply demand ratio should be balanced across the whole city, and one objective of the dispatch problem is to minimize the following function
\begin{align}
J_E=\sum_{k=1}^{T}\left\|\frac{1}{N}\mathbf{1}^T_N X^k - \frac{1}{R^k}r^k\right\|_1.
\label{JE}
\end{align}

The other objective is to reduce total idle driving distance of all taxis. The process of traversing from the initial location to the dispatched region will introduce an idle driving distance for a vacant taxi, and we consider to minimize such idle driving distance associated with meeting the dispatch solutions. 

We begin with estimate the total idle driving distance associated with meeting the dispatch solutions. For the convenience of routing process,  the dispatch center is required to send the GPS location of the destination to vacant taxis. The decision variable $X^k$ only provides the region each vacant taxi should go, hence we map the region ID to a specific longitude and latitude position for every taxi. In practice, there are taxi stations on roads in metropolitan areas, and we assume that each taxi has a preferred station or is randomly assigned one at every region by the dispatch system. 
We denote the preferred geometry location matrix for the $i$-th taxi by $W_{i} \in \mathbb{R}^{n\times 2}$, and $[W_i]_j$, where each row of $W_i$ is a two-dimensional geometric position on the map. The $j$-th row of $W_i$ is the dispatch position sent to the $i$-th taxi when $X^k_{ij}=1$.  

Once $X_i^k$ is chosen, then the $i$-th taxi will go to the location $X_i^k W_i$,  because the following equation holds
\\\centerline{$
%\begin{align*}
X^k_i W_i=\sum_{q\neq j}X^k_{iq}[W_i]_q+X^k_{ij}[W_i]_j=[W_i]_j \in \mathbb{R}^{1\times 2}.
%\end{align*}
$}
With a binary vector $X^k_i$ that $X^k_{ij}=1$, $X^k_{iq}=0$ for $q \neq j$, we have $X^k_{iq}W_i=[0\ 0]$ for $q \neq j$.  Since $W_i$ does not need to change with time $k$, the preferred location of each taxi at every region in the city is stored as a matrix $\mathbf{W}$, stored in the dispatch center before the process of calculating dispatch solutions starts. When updating information of vacant taxis, matrix $W_i$ is also updated for every current vacant taxi $i$.

The initial position $P^0_i$ is provided by GPS data. Traversing from position $P^0_i$ to position $X^1_i W_i$ for predicted demand will introduce a cost, since the taxi drives towards the dispatched locations without picking up a passenger. Hence, we consider minimizing the total idle driving distance introduced by dispatching taxis. 
Driving in a city is approximated as traveling on a grid road. To estimate the distance without knowing the exact path, we use the Manhattan norm or one norm between two geometric positions, which is widely applied as a heuristic distance in path planning algorithms~\cite{manhattan}. We define $d_i^k \in \mathbb{R}$ as the estimated idle driving distance of the $i$-th taxi for reaching the dispatched location $X_i^k W_i$. For $k=1$, a lower bound of $d^1_i$ is given by
\begin{align}
d^1_i \geqslant \|P^0_i-X^1_iW_i\|_1, \quad i=1,\dots, N.
\label{d0}
\end{align}

For $k\geqslant 2$, to estimate the anticipated future idle driving distance induced by reaching dispatched position $X_i^{k} W_i$ at time $k$, we consider the trip at the beginning of time slot $k$ starts at the end location of time slot $k-1$. However, during time $k-1$, taxis' mobility patterns are related to pick-up and drop-off locations of passengers, which are not directly controlled by the dispatch center. So we assume the predicted ending position for a pick-up location $X^{k-1}_iW_i$ during time $k-1$ is related to the starting position $X^{k-1}_iW_i$ as follows:\begin{align}
P_i^{k-1}=f^k(X^{k-1}_iW_i),\quad f^k: \mathbb{R}^{1\times 2} \to \mathbb{R}^{1\times 2},
\label{Pk}
\end{align}
where $f^k$ is a convex function, called a mobility pattern function. To reach the dispatched location $X_i^k W_i $ at the beginning of time $k$ from position $P^{k-1}_i$, the approximated driving distance is %estimated by the Manhattan norm
%\footnotesize
\begin{align}
d^k_i \geqslant \| f^k(X^{k-1}_iW_i)-X^{k}_i W_i\|_1. %\ k= 2,\cdots, T, \ i =1,\dots, N.
\label{uk}
\end{align}
The process to calculate a lower bound for $d_i^k$ is illustrated in Figure~\ref{u}. 

\begin{figure}[b!]
\vspace{-8pt}
\centering
\includegraphics [width=0.4\textwidth]{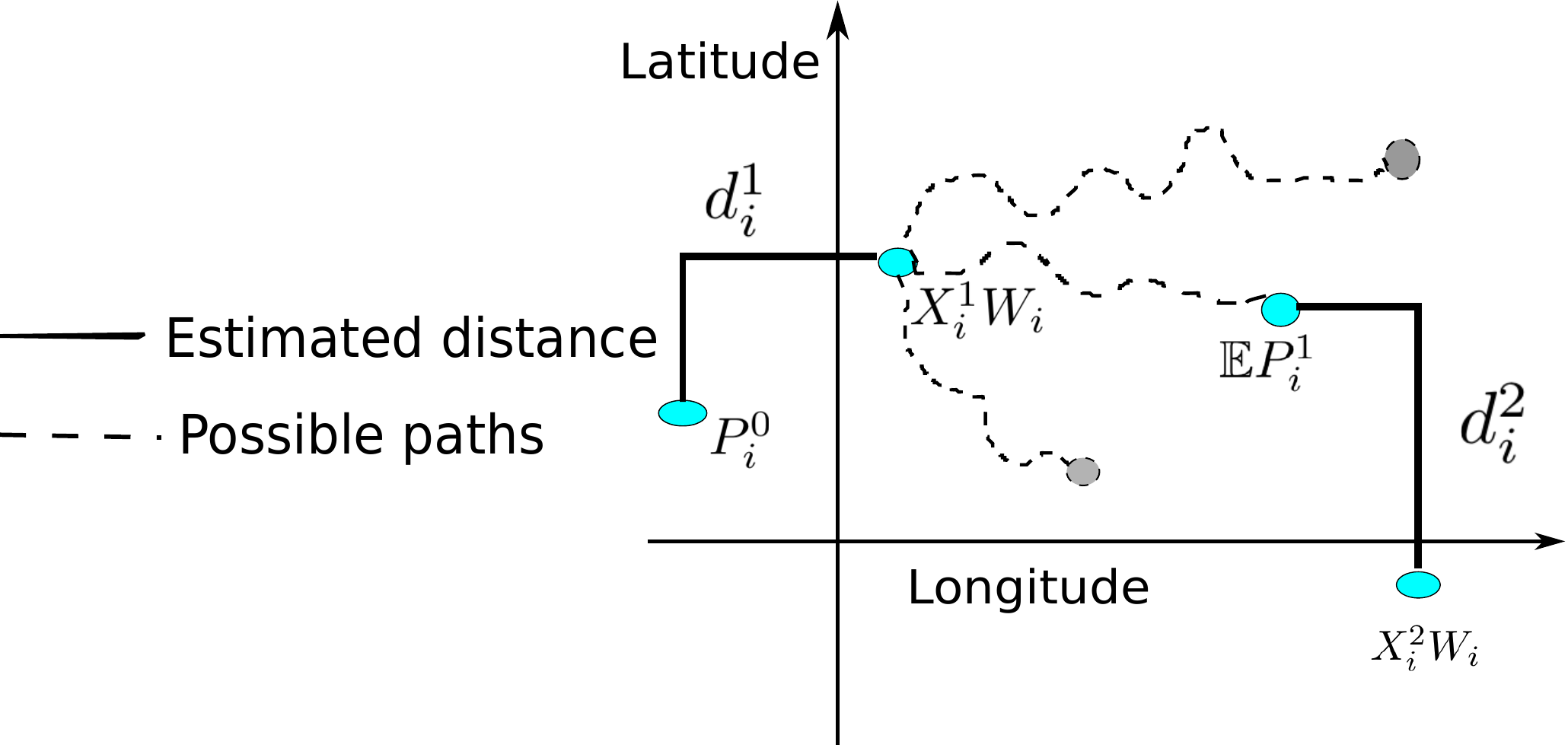}
\vspace{-5pt}
\caption{Illustration of the process to estimate idle driving distance to the dispatched location for the $i$-th taxi at $k=2$: predict ending location of $k=1$ denoted by $\mathbb{E}P^1_i$ in~\eqref{EP}, get the distance between locations $\mathbb{E}P^1_i$ and $X^2_i W_i$ denoted by $d^2_i$ in~\eqref{uk}.}
\label{u}
%\vspace{-8pt}
\end{figure}
%\normalsize
% is a convex function of $X_i^{k-1}$ and $X_i^k$, and 

Within time slot $k$, the distance that every taxi can drive should be bounded by a constant vector $\alpha^k \in \mathbb{R}^{N}$:
\\\centerline{$
%\begin{align*}
%\vspace{-5pt}
d^k  \leqslant \alpha^k.
%\end{align*}
$}
Total idle driving distance of all vacant taxis though time $k=1,\dots, T$ to satisfy service fairness is then denoted by
\begin{align}
J_D=\sum\limits_{k=1}^{T}\sum_{i=1}^{N} d_i^k.
\label{Jd}
\end{align}

It is worth noting that the idle distance we estimate here is the region-level distance to pick up predicted passengers --- the distance is nonzero only when a vacant taxi is dispatched to a different region. We also require that the estimated distance is a closed form function of the locations of the original and dispatched regions, without knowledge about accurate traffic conditions or exact time to reach the dispatched region. Hence, in this work we use Manhattan norm to approximate the idle distance---it is a closed form function of the locations of the original and dispatched regions. When accessibility information of the road traffic network is considered in estimating street-level distances, for the case that a taxi may not drive on rectangular grids to pick up a passenger (for instance, when a U-turn is necessary), Lee et.al\ have proposed a shortest time path approach to pick up passengers in shortest time~\cite{Lee_review}.
%One basic requirement for a dispatch solution is to save idle driving cost. The variable measures driving distance is $u^k$, and we define a quadratic form of idle driving energy cost as $(u^k)' u^k$. 

%To reduce extra idle distance of future time slots introduced by the dispatch solution of the current period, we consider total costs in several consecutive time slots instead of one single time slot.
\iffalse
Since there exists a trade-off between two objectives as discussed in Section~\ref{sec:prob_descrip}, we define a weight parameter $\beta$ when summing up the costs related to both objectives, and the cost function $J$ is
\begin{align}
J_E+\beta J_D=\sum\limits_{k=1}^{T}\left(\left\|\frac{1}{N}\mathbf{1}^T_N X^k - \frac{1}{R^k}r^k\right\|_1+\beta \sum_{i=1}^{N} d_i^k\right).
\label{obj}
\end{align}
\fi
\subsubsection{An RHC problem formulation}%{A multi-objective optimization problem formulation}
Since there exists a trade-off between two objectives as discussed in Section~\ref{sec:prob_descrip}, we define a weight parameter $\beta^k$ when summing up the costs related to both objectives.
A list of parameters and variables is shown in Table~\ref{opt_parameter}. When mixed integer programming is not efficient enough for a large-scale taxi network regarding to the problem size, one standard relaxation method is replacing the constraint $X^k_{ij} \in \{0,1\}$ by $0 \leq X^k_{ij} \leq 1.$
%\begin{align*}
%\begin{split}
%\begin{center}

%\forall k \in \{1, 2 \dots, T\}, \\ &i \in \{1, 2, \dots, N \}, j\in\{1,2, \dots, n\}.
%\end{split}
%\vspace{-8pt}
%\label{relax}
%\end{align*}

%and the cost function $J$ is
To summarize, we formulate the following problem~\eqref{opt2} based on the definitions of variables, parameters, constraints and objective function
\begin{align}
\begin{split}
\underset{X^k, d^k}{\text{min.}}\quad& J= J_E+\beta J_D   \\
&\quad= \sum\limits_{k=1}^{T}\left(\left\|\frac{1}{N}\mathbf{1}^T_N X^k - \frac{1}{R^k}r^k\right\|_1+\beta^k \sum_{i=1}^{N} d_i^k\right)\\
\text{s.t}\quad %X^1= P^1\\
&d^1_i\geqslant \|P^0_i-X^1_iW_i\|_1,\quad i=1,\dots,N,\\
%&P^k= X^k C^k, k=1,\cdots, T-1,\\
&d^{k}_i\geqslant \|f^k(X^{k-1}_i W_i) -X^{k}_iW_i\|_1, \\
&i=1,\dots, N, \quad k= 2,\dots, T, \\%i=1,2,\dots, N,\\
%(or\ P^k_i=C^k X^k_i),\\
& d^k \leqslant \alpha^k, \quad k=1,2, \dots, T,\\
& X^k \mathbf{1}_n=\mathbf{1}_N, \quad k=1, 2, \dots, T,\\
%&\frac{\mathbf{1}_N X^k}{\mathbf{r}_k} =\frac{\mathbf{1}_N X^k \mathbf{1}_n}{\mathbf{r}_k \mathbf{1}_n},k=2,\cdots, T,\\
&0 \leqslant X^k_{ij} \leqslant 1, i\in\{1,\dots,N\},  j \in \{1, \dots, n\}.
%&0 \leq X^k_{ij} \leq 1, \forall k \in \{1, 2 \dots, T\}, \\
 %\quad j\in\{1,2, \dots, n\}.
%&\mathbf{1}_N X^k \geqslant c\mathbf{r}_k, k=1,\cdots, T
\end{split}
\label{opt2}
\end{align}
After getting an optimal solution $X^1$ of problem~\eqref{opt2}, for the $i$-th taxi, we may recover binary solution through rounding by setting the largest value of $X^1_{i}$ to $1$, and the others to $0$. This may violate the constraint of $d_i^0$, but since we set a conservative upper bound $\alpha^k$, and the rounding process will return a solution that satisfies $d^k_i \leqslant \alpha^k + \epsilon$ with bounded $\epsilon$, the dispatch solution can still be executed during time slot $k$.

\subsection{Discussions on the optimal dispatch formulation}

\subsubsection{Why use supply demand ratio as a metric}
An intuitive measurement of the difference between the number of vacant taxis and predicted total requests at all regions is:
\\\centerline{$
%\begin{align*}
e=\sum\limits_{j=1}^{n} |s^k_j-r^k_j|,
%\end{align*}
$}
where $s^k_j$ is the total number of vacant taxis sent to the $j$-th region. However, when the total number of vacant taxis and requests are different in the city, this error $e$ can be large even under the case that more vacant taxis are already allocated to busier regions and fewer vacant taxis are left to regions with less predicted demand. We do not have an evidence whether the dispatch center already fairly allocates supply according to varying demand given the value of the above error $e$.

\subsubsection{The meaning of $\alpha^k$}
\label{alphak} 
For instance, when the length of time slot $k$ is one hour, and $\alpha^k$ is the distance one taxi can traverse during $20$ minutes of that hour, this constraint means a dispatch solution involves the requirement that a taxi should be able to arrive the dispatched position within $20$ minutes in order to fulfill predicted requests. With traffic condition monitoring and traffic speed predicting method~\cite{Jaillet}, $\alpha^k$ can be adjusted according to the travel time and travel speed information available for the dispatch system. This constraint also gives the dispatch system the freedom to consider the fact that drivers may be reluctant to drive idly for a long distance to serve potential customers, and a reasonable amount of distance to go according to predicted demand is acceptable. The threshold $\alpha^k$ is related to the length of time slot. In general, the longer a time slot is, the larger $\alpha^k$ can be, because of constraints like speed limit.

\subsubsection{One example of mobility pattern function $f^k$} When taxi's mobility pattern during time slot $k$ is described by a matrix $C^k \in \mathbb{R}^{n\times n}$ satisfying $\sum_{j=1}^{n} C_{ij}=1$,
where $C^k_{ij}$ is the probability that a vacant taxi starts within region $i$ will end within region $j$ during time $k$. From the queueing-theoretical perspective such probability transition matrix approximately describes passenger's mobility~\cite{mod}. 
Given $X^{k-1}_i$ and the mobility pattern matrix $C^{k-1} \in [0,1]^{n \times n}$, the probability of ending at each region for taxi $i$ is
\begin{align*}
p=\sum_{j=1}^{n}[C^{k-1}]_jI(X^{k-1}_{ij}=1)=X^{k-1}_iC^{k-1} \in \mathbb{R}^{1\times n},
\end{align*}
where the indicator function $I(X^{k-1}_{ij}=1)=1$ if and only if $X^{k-1}_{ij}=1$, and $[C^{k-1}]_j$ is the $j$-th row of $C^{k-1}$.  However, introducing a stochastic formula in the objective function will cause high computational complexity for a large-scale problem. Hence, instead of involving the probability of taking different paths in the objective function to formulate a stochastic optimization problem, we take the expected value of the position of $i$-th taxi by the end of time $k-1$ 
\begin{align}
P^{k-1}_i=\sum_{j=1}^{n}p_j[W_i]_j=pW_i=X_i^{k-1} C^{k-1} W_i.
\label{EP}
\end{align} 
Here $P^{k-1}_i \in \mathbb{R}^{1 \times 2}$ is a vector representing  a predicted ending location of the $i$-th taxi on the map at each dimension. Then a lower bound of idle driving distance for heading to $X_i^{k} W_i$ to meet demand during $k$ is given by the distance between $P^{k-1}_i$ defined as~\eqref{EP} and $X_i^{k} W_i$.
\begin{align}
d^k_i \geqslant \|(X_i^{k-1}C^{k-1} -X^{k}_i)W_i\|_1. %\ k= 2,\cdots, T, \ i =1,\dots, N.
\label{uk}
\end{align}
In particular, when the transition probability $C^{k},k=1,\dots, T$ is available, we can replace the constraint about $d_i^k$ by $d^{k}_i \geqslant \|(X^{k-1}_i C^{k-1} -X^{k}_i)W_i\|_1$. 
 
It is worth noting that $d^k_i$ is a function of $X_i^{k-1}$ and $X_i^k$, and the estimation accuracy of idle driving distance to dispatched positions $X_i^k$ ($k=2, \dots, T$) depends on the predicting accuracy of the mobility pattern during each time slot $k$, or $P^{k-1}_i$. The distance $d^1$ is calculated based on real-time GPS location $P^0$ and dispatch position $X^1$, and we use estimations $d^{2},\dots,d^{T}$ to measure the anticipated future idle driving distances for meeting requests.

The error of estimated $C^k$ mainly affects the choice of idle distance $d^k$ when the true ending region of a taxi by the end of time slot $k$ is not as predicted based on its starting region at time slot $k$. This is because $C^k$ determines the constraint for $d^k$ ($k=2,3,\dots, T$) as described by inequality~\eqref{uk}. However, the system also collects real-time GPS positions to make a new decision based on the current true positions of all taxis, instead of only applying predicted locations provided by the mobility pattern matrix. According to constraint~\eqref{d0} distance $d^1$ is determined by GPS sensing data $P^0$ and dispatch decision $X^1$, and only  $X^1$ will be executed sent to vacant taxis as the dispatch solutions after the system solving problem~\eqref{opt2}. From this perspective, real-time GPS and occupancy status sensing data is significant to improve the system's performance when we utilize both historical data and real-time sensing data. We also consider the effect of an inaccurate mobility pattern estimation $C^k$ when choosing the prediction time horizon $T$ --- large prediction horizon will induce accumulating prediction error in matrix $C^k$ and the dispatch performance will even be worse. Evaluation results in Section~\ref{sec:simulation} show how real-time sensing data helps to reduce total idle driving distance and how the mobility pattern error of different prediction horizon $T$ affects the system's performance. 

% Do not talk about $T$ here! There will be discussion in Algorithm section.

\subsubsection{Information on road congestion and passenger destination}
When road congestion information is available to the dispatch system, function in~\eqref{Pk} can be generalized to include real-time congestion information. For instance, there is a high probability that a taxi stays within the same region during one time slot under congestions. 
%However, road congestion is difficult to predict for future time in general,  hence, in this work, we do not assume available road congestion model for future time slots $k=2,\dots, T$, and our experiments do not involve any mobility patterns under congestions.

We do not assume that information of passenger's destination is available to the system when making dispatch decisions, since many passengers just hail a taxi on the street or at taxi stations instead of reserving one in advance in metropolitan areas. When the destination and travel time of all trips are provided to the dispatch center via additional software or devices as prior knowledge, the trip information is incorporated to the definition of ending position function~\eqref{Pk} for problem formulation~\eqref{opt2}. With more accurate trip information, we get a better estimation of future idle driving distance when making dispatch decisions for $k=1$.

\subsubsection{Customers' satisfaction under balanced supply demand ratio}
%According to the design, there exists situation that taxi $i$ will not pick up passengers in its original region but will be dispatched to another region. This will lead to increase in passengers waiting time.  such that customer's average waiting time is almost the same at different regions, Nevertheless, 
The problem we consider in this work is reaching fair service to increase global level of customers' satisfaction, which is indicated by a balanced supply demand ratio across different regions of one city, instead of minimizing each individual customer's waiting time when a request arrives at the dispatch system. Similar service fairness metric has been applied in mobility on demand systems~\cite{mod}, and supply demand ratio considered as an indication of utilization ratio of taxis is also one regulating objective in taxi service market~\cite{Yang_review}. For the situation that taxi $i$ will not pick up passengers in its original region but will be dispatched to another region, the dispatch decision results from the fact that global customers' satisfaction level will be increased. For instance, when the original region of taxi $i$ has a higher supply demand ratio than the dispatched region, going to the dispatched region will help to increase customer's satisfaction in that region. By sending taxi $i$ to some other region, customers' satisfaction in the dispatched region can be increased, and the value of the supply-demand cost-of-mismatch function $J_E$ can be reduced without introducing much extra total idle driving distance $J_D$.

\subsection{Robust RHC formulations}
Previous work has developed multiple ways to learn passenger demand and taxi mobility patterns~\cite{Jaillet,time_predict,Huang:2012}, and accuracy of the predicted model will affect the results of dispatch solutions. We do not have perfect knowledge of customer demand and taxi mobility models in practice, and the actual spatial-temporal profile of passenger demands can deviate from the predicted value due to random factors such as disruptive events. Hence, we discuss formulations of robust taxi dispatch problems based on~\eqref{opt2}.

Formulation~\eqref{opt2} is one computationally tractable approach to describe the design requirements with a nominal model. One advantage of the formulation~\eqref{opt2} is its flexibility to adjust the constraints and objective function according to different conditions. With prior knowledge of scheduled events that disturb the demand or mobility pattern of taxis, we are able to take the effects of the events into consideration by setting uncertainty parameters.
For instance, when we have basic knowledge that total demand in the city during time $k$ is about $\tilde{R}^k$, but each region $r^k_j$ belongs to some uncertainty set, denoted by an entry wise inequality
\\\centerline{$
%\begin{align*} 
R^k_{1} \preceq r^k \preceq R^k_2,
%\end{align*} 
$}
given $R^k_1\in \mathbb{R}^{n}$ and $R^k_2\in\mathbb{R}^n$. Then 
\begin{align}
r_j^k \in [R^k_{1j},R^k_{2j}], j=1,\dots, n
\label{r_jk_uncertainty}
\end{align}
is an uncertainty parameter instead of a fixed value as in problem~\eqref{opt2}. 
Without additional knowledge about the change of total demand in the whole city, we denote $\tilde{R}^k$ as the approximated total demand in the city under uncertain $r^k_j$ for each region. By introducing interval uncertainty~\eqref{r_jk_uncertainty} to $r^k$ and fixing $\tilde{R}^k$ on the denominator, we have the following robust optimization problem~\eqref{robust_mpc}
\begin{align}
\begin{split}
\underset{X^k, d^k}{\text{min.}}\quad& \underset{R^k_{1} \preceq r^k \preceq R^k_2}{\text{max}}\sum\limits_{k=1}^{T}\left(\left\|\frac{1}{N}\mathbf{1}^T_N X^k - \frac{1}{\tilde{R}^k}r^k\right\|_1+\beta^k\sum_{i=1}^{N} d_i^k\right)\\
\text{s.t.}\quad &\text{constraints of problem~\eqref{opt2}}.
%&\mathbf{1}_N X^k \geqslant c\mathbf{r}_k, k=1,\cdots, T
\end{split}
\label{robust_mpc}
\end{align}
The robust optimization problem~\eqref{robust_mpc} is computationally tractable, and we have the following Theorem~\ref{solve_robust} to show the equivalent form to provide real-time dispatch decision.
% \footnote{For a complete proof, please see the attachment, or full version with appendix on https://sites.google.com/site/miaofeiatpenn/publications.}
\begin{theorem}
The robust RHC problem~\eqref{robust_mpc} is equivalent to the following computationally efficient convex optimization problem
\begin{align}
\begin{split}
\underset{X^k, d^k, t^k}{\mathrm{min}}\quad &J'=\sum\limits_{k=1}^{T}\left(\sum_{j=1}^{n} t_j^k+\beta^k\sum_{i=1}^{N} d_i^k\right)\\
\mathrm{s.t}\quad &t^k_j\geq \frac{\mathbf{1}_N X_{\cdot j}^k}{N}-\frac{R^k_{1j}}{\tilde{R}^k}, 
                                        t^k_j \geq \frac{R^k_{1j}}{\tilde{R}^k}-\frac{\mathbf{1}_N X_{\cdot j}^k}{N},\\
                                        &t^k_j \geq \frac{\mathbf{1}_N X_{\cdot j}^k}{N}- \frac{R^k_{2j}}{\tilde{R}^k},
                                         t^k_j \geq  \frac{R^k_{2j}}{\tilde{R}^k}-\frac{\mathbf{1}_N X_{\cdot j}^k}{N},\\ 
                                        &j=1,\dots, n,\quad, k=1,\dots,T,\\
                                                                                                                     \iffalse
\geq &\mathrm{max}\{\frac{\mathbf{1}_N X_{\cdot j}^k}{N}-\frac{R^k_{1j}}{\tilde{R}^k},\frac{R^k_{1j}}{\tilde{R}^k}-\frac{\mathbf{1}_N X_{\cdot j}^k}{N}, \\
&\quad\quad \frac{\mathbf{1}_N X_{\cdot j}^k}{N}- \frac{R^k_{2j}}{\tilde{R}^k},\frac{R^k_{2j}}{\tilde{R}^k}-\frac{\mathbf{1}_N X_{\cdot j}^k}{N}\}\\
\fi
&\mathrm{constraints\ of\ problem~\eqref{opt2}}.
\end{split}
\label{robust_mpc_conv}
\end{align}
\label{solve_robust}
\end{theorem}
%\mathbf{1}^T_n r^k
\begin{proof}
In the objective function, only the first term is related to $r^k$. To avoid the maximize expression over an uncertain $r^k$, we first optimize the term over $r^k$ for any fixed $X^k$. Let $X^k_{\cdot j}$ represent the $j$-th column of $X^k$, then
\begin{align}
\begin{split}
&\underset{R^k_{1} \preceq r^k \preceq R^k_2}{\text{max}}\left\|\frac{1}{N} \mathbf{1}^T_N X^k- \frac{1}{\tilde{R}^k}r^k\right\|_1 \\
=&\underset{R^k_{1} \preceq r^k \preceq R^k_2}{\text{max}}\sum_{j=1}^{n}\left|\frac{1}{N} \mathbf{1}^T_N X^k_{\cdot j}- \frac{r_j^k}{\tilde{R}^k}\right|\\
=& \sum_{j=1}^{n}\underset{r_j^k \in [R^k_{1j},R^k_{2j}]}{\text{max}}\left|\frac{1}{N} \mathbf{1}^T_N X^k_{\cdot j}- \frac{r_j^k}{\tilde{R}^k}\right|. 
\end{split}
\label{max}
\end{align}
The second equality holds because each $r^k_j$ can be optimized separately in this equation.
For $R^k_{1j}\leq r^k_j \leq R^k_{2j}$, we have 
\begin{align*}
\frac{R^k_{1j}}{\tilde{R}^k}\leq \frac{r^k_j}{\tilde{R}^k} \leq \frac{R^k_{2j}}{\tilde{R}^k}.
\end{align*}
Then the problem is to maximize each absolute value in~\eqref{max} for $j=1,\dots, n$. Consider the following problem for $x,a, b\in\mathbb{R}$ to examine the character of maximization problem over an absolute value:
\begin{align*} 
\underset{x_0 \in [a,b]}{\text{max}}|x-x_0|
&=\begin{cases} |x-a|,\quad \text{if}\  x>(a+b)/2\\
                               |x-b|,\quad \text{otherwise}\end{cases}\\
                                                         &=\text{max}\{|x-a|,|x-b|\}\\
                                                                         &=\text{max}\{x-a,a-x,x-b,b-x\}.
\end{align*} 

Similarly, for the problem related to $r^k_j$, we have
\begin{align}
\begin{split}
&\underset{r_j^k \in [R^k_{1j},R^k_{2j}]}{\text{max}}\left|\frac{\mathbf{1}_N X_{\cdot j}^k}{N} - \frac{r_j^k}{\tilde{R}^k}\right|\\ 
\quad= &\text{max}\left\{\left|\frac{\mathbf{1}_N X_{\cdot j}^k}{N}-\frac{R^k_{1j}}{\tilde{R}^k}\right|,\quad\left|\frac{\mathbf{1}_N X_{\cdot j}^k}{N}- \frac{R^k_{2j}}{\tilde{R}^k}\right|\right\}. 
\end{split}
\label{rk}
\end{align}
\iffalse
If we sum up the right hand side of equation~\eqref{rk} for every $j$, we get an upper bound of 
$\sum_{j=1}^{n}\underset{r^k \in [R^k_{1},R^k_{2}]}{\text{max}}\left|\frac{\mathbf{1}_N X_{\cdot j}^k}{N} - \frac{r_j^k}{\mathbf{1}_n r^k }\right|$, since $\mathbf{1}_n r^k$ depends on every value of $r^k_j$ and the maximum value of each right hand side of equation~\eqref{rk} may not be reached at the same time. 
\fi
Thus, with slack variables $t^k \in \mathbb{R}^{n}$, 
we re-formulate the robust RHC problem as~\eqref{robust_mpc_conv}.
\end{proof}

Taxi mobility patterns during disruptive events can not be easily estimated (in general), however, we have knowledge such as a rough number of people are taking part in a conference or competition, or even more customer reservations because of events in the future. The uncertain set of predicted demand $r^k$ can be constructed purely from empirical data such as confidence region of the model, or external information about disruptive events. 
By introducing extra knowledge besides historical data model, the dispatch system responds to such disturbances with better solutions than the those without considering model uncertainties. Comparison of results of~\eqref{robust_mpc_conv} and problem~\eqref{opt2} is shown in Section~\ref{sec:simulation}. 

%%%%%%%%%%
%%%%%%%%%
%%%%%%%%%%%%%%%

\section{RHC Framework Design}
\label{sec:algorithm}
Demand and taxi mobility patterns can be learned from historical data, but they are not sufficient to calculate a dispatch solution with dynamic positions of taxis when the positions of the taxis change in real time.  
%-- we need to adjust dispatch solutionswhen moving the time horizon. 
%time-variant conditions. -- we need to specify a region for each vacant taxi according to its latest location. 
Hence, we design an RHC framework to adjust dispatch solutions according to real-time sensing information in conjunction with the learned historical model. Real-time GPS and occupancy information then act as feedback by providing the latest taxi locations, and demand-predicting information for an online learning method like~\cite{Dmodel}.
Solving problem~\eqref{opt2} or~\eqref{robust_mpc} is the key iteration step of the RHC framework to provide dispatch solutions. %is provided by solving~\eqref{opt2},~\eqref{robust_mpc}. %according to taxi locations and predicted requests to be served by vacant taxis,    
% that dispatches vacant taxis according to their current positions, based on both historical and real-time information. 
%Historical data provides estimated parameters, and 
%When moving the time horizon ahead at each iteration, real-time information helps

%The MPC framework updates parameters for~\eqref{opt2} during each iteration, %We design an MPC taxi dispatch algorithm, to utilize both historical and real-time GPS and occupancy data -- 
\iffalse
\subsubsection{MPC Algorithm for taxi dispatch}
\begin{figure}[b!]
%\vspace{-5pt}
\centering
\includegraphics [width=0.43\textwidth]{mpc_basic.pdf}
\vspace{-8pt}
\caption{The process of Algorithm~\ref{mpc_rt}: each step, update demand, calculate solutions of problem~\eqref{opt2} or~\eqref{robust_mpc}, send dispatch positions according to $X^1$, and move to the next time horizon.}
\label{mpc_basic}
\end{figure}
\fi
%MPC is a receding horizon control method widely applied in control literature.%~\cite{modelpc2013}. 
%At each time step, the system model and current measurement are used to predict the future outputs, and a sequence of control signal is calculated to minimize the objective function. Then the actuator only implements the first control signal of the sequence.
\iffalse At each iteration, we update GPS data and occupancy status of all taxis, estimate remaining requests of the following time slots, such that the dispatch order calculated by~\eqref{opt2} or~\eqref{robust_mpc}  reflects the current supply and demand distribution. \fi

RHC works by solving the cost optimization over the window $[1,T]$ at time $k=1$. Though we get a sequence of optimal solutions in $T$ steps -- ${X}^1, \dots, X^T$, we only send dispatch decisions to vacant taxis according to $X^1$. We summarize the complete process of dispatching taxis with both historical and real-time data as Algorithm~\ref{mpc_rt}, followed by a detail computational process of each iteration. The lengths of time slots for learning historical models ($t_1$) and updating real-time information ($t_2$) do not need to be the same, hence in Algorithm~\ref{mpc_rt} we consider a general case for different $t_1,t_2$.
%Models learned offline based on historical data before implementing the algorithm are considered as inputs. %as dispatch orders. 
%The algorithm process is shown in Figure~\ref{mpc_basic}. 
\subsection{RHC Algorithm}
\begin{algorithm}%[H]
\caption{RHC Algorithm for real-time taxi dispatch}
%\begin{algorithm}
\textbf{Inputs:} 
%A system model and multi-objective optimization parameters: demand model 
{Time slot length $t_1$ minutes, period of sending dispatch solutions $t_2$ minutes ($t_1/t_2$ is an integer); %Set the values of $t_2, t_1$ to make sure that . 
a preferred station location table $\mathbf{W}$ for every taxi in the network; %\SHfull{Do you have a math symbol for this?}
estimated request vectors $\hat{r}(h_1)$, $h_1=1,\dots,1440/t_1$,
%\protect\footnote{if online training is available, this offline learning model $\hat{r}(h_1)$ is not necessary.},
mobility patterns $\hat{f}(h_2)$, $h_2=1,\dots,1440/t_2$; prediction horizon $T \geq 1$.
%parameters of problems~\eqref{opt2}, \eqref{robust_mpc}: prediction horizon $T \geq 1$,  $\beta$, and $\alpha$.
}
%Store the estimation results. %as given parameters for the following iterations. %Let $l=\frac{t_1}{t_2}$.   
\\\textbf{Initialization:}
{%Pick the corresponding $\hat{r}(h_1), \hat{f}(h_2)$ from the inputs, according to the start time of the iteration. 
%and assume we start at the beginning of the $h_2$-th time slot of length $t_2$. Initialize 
The predicted requests vector $r=\hat{r}(h_1)$ for corresponding algorithm start time $h_1$.}\\
%\\\textbf{Iteration}: 
\While{Time is the beginning of a $t_2$ time slot} {
 %For each larger time slot $t_1$, divide it into $l$ shorter sub time slots, 
%(1). Collect GPS and occupancy information, update parameters:  
%\SHfull{This is very confusing, since ``while'' must be followed by a true/false statement.}
(1) Update sensor information for initial position of vacant taxis $P^0$ and occupied taxis $P'^0$, total number of vacant taxis $N$, preferred dispatch location matrices $W_i$.\\%i=1,\dots,N$.\\
\If{time is the beginning of an $h_1$ time slot}
{Calculate $\hat{r}(h_1)$ if the system applies an online training method; count total number of occupied taxis $n_o(h_1)$; update vector $r$.}
%{update pick up and drop off information of the past period for online learning; .}\\
%%total number of occupied taxis at $h_1$.}\\
% waiting requests vectors for vacant taxis in the following $T$ time slots.\\
(2) Update the demand vectors $r^k$ based on predicted demand $\hat{r}(h_1)$ and potential service ability of $n_o(h_1)$ occupied taxis; update mobility functions $f^k(\cdot)$ (for example, $C^k$), set up values for idle driving distance threshold $\alpha^k$ and objective weight $\beta^k$, $k=1,2,\dots, T$. %,$=\hat{C}(h_2+k-1), k=1,\dots, T$.\\
%For $k=1$ to $k=T$, define predicted requests that will be served by current occupied taxis in next time slot as $r_o$, $r^k=r-r_o$,  \\
%
(3) %Compute variables $X^k, k=1,2,\dots,T$.\\
%of problem~\eqref{opt2} or problem~\eqref{robust_mpc}, with the parameters from step (1):\\
\eIf{there is knowledge of demand $r^k$ as an uncertainty set ahead of time}
   {solve problem~\eqref{robust_mpc_conv};}
   {solve problem~\eqref{opt2} for a certain demand model\;}
(4) Send dispatch orders to vacant taxis according to the optimal solution of matrix $X^1$. Let $h_2=h_2+1$.}
\textbf{Return:}{Stored sensor data and dispatch solutions.} %and parameter adjustment.
\label{mpc_rt}
%\end{algorithm}
\end{algorithm}
\begin{remark}
Predicted values of requests $\hat{r}(h_1)$ depend on the modeling method of the dispatch system. For instance, if the system only applies historical data set to learn each $\hat{r}(h_1)$, $\hat{r}(h_1)$ is not updated with real-time sensing data. When the system applies online training method such as~\cite{Dmodel} to update  $\hat{r}(h_1)$ for each $h_1$, values of $r$, $r^k$ are calculated based on the real-time value of $\hat{r}(h_1)$.
%we first update based on the modeling method,  
\end{remark}
%%%%%%%%%%%%%
%%%%%%%%%%%%%%%
%%%%%%%%%%%%%%%
%\vspace{-8pt}
%%%%%%%%%%%%
%%%%%%%%%%%%%
\subsubsection{Update $r$}
We receive sensing data of both occupied and vacant taxis in real-time. 
\iffalse If real-time training algorithm to learn the requests is available, the input of model $\hat{r}(h_1)$ learned based only on historical data is not necessary.
To assign values of $r^k$ in the iteration step (2), \fi
Predicted requests that vacant taxis should serve during $h_1$ is re-estimated at the beginning of each $h_1$ time. %if the time slot $t_2$ is also the beginning of a time slot $t_1$, let
%Besides taxis vacant at the beginning of time $h_1$, taxis with passengers may drop off passengers and be able to serve at a region.  
 To approximate the service capability when an occupied taxi turns into vacant during time $h_1$, we define the total number of drop off events at different regions as a vector $dp(h_1) \in \mathbb{R}^{n\times1}$.  Given $dp(h_1)$, the probability that a drop off event happens at region $j$ is %$j$ during time slot $h_1$ is
\begin{align}
pd_j(h_1)=dp_j(h_1)/ \mathbf{1}_n dp(h_1),
\label{pd}
\end{align} 
where $dp_j(h_1)$ is the number of drop off events at region $j$ during $h_1$.  %We assume during time slot $h_1$, 
%the probability that an occupied taxi of region $j$ turns into vacant at region $j$ is approximated by $pd_j(h_1)$. 
We assume that an occupied taxi will pick up at least one passenger within the same region after turning vacant, and we approximate future service ability of occupied taxis at region $j$ during time $h_1$ as
\\\centerline{$
%\begin{align}
r_{oj}(h_1)=\ceil*{pd_j(h_1)\times n_{o}(h_1)},
%\label{ro}
%\end{align} 
$}
where $\ceil*{\cdot}$ is the ceiling function, $n_{o} (h_1)$ is the total number of current occupied taxis at the beginning of time $h_1$ provided by real-time sensor information of occupied taxis. Let %$P'^0$. %(updated every $t_2$ period). 
%%%%%%%%%%%%%%%
%To assign $r^k$ in the iteration step (3), we update predicted requests vacant taxis should satisfy at the beginning of each time slot $h_1$: if the time slot $t_2$ is also the beginning of a time slot $t_1$, let
%%%%%%%%%%%%%%%%%
\\\centerline{$
%\begin{align}
r=\hat{r}(h_1)-r_o(h_1),
%\label{r}
%\end{align} 
$}
then the estimated service capability of occupied taxis is deducted from $r$ for time slot $h_1$. 

\subsubsection{Update $r^k$ for problem~\eqref{opt2}}
We assume that requests are uniformly distributed during $h_1$. 
Then for each time $k$ of length $t_2$, if the corresponding physical time is still in the current $h_1$ time slot,   the request is estimated as an average part of $r$; else, it is estimated as an average part for time slot $h_1+1, h_1+2, \dots$, etc. The rule of choosing $r^k$ is %described by
\begin{align*}
r^k=\begin{cases}\frac{1}{H} r, \ &\text{if}\ (k+h_2-1)t_2\leq h_1t_1\\
                                          \frac{1}{H}\hat{r}\left(\ceil*{\frac{(k+h_2-1)t_2}{t_1}}\right), \ &\text{otherwise}\end{cases}
\end{align*}
where $H=t_1/t_2$.
%For $k$ such that $t_2*k\leq t_1$, $$; else, $r^k=$,
%where  
%where $H=\frac{t_2}{ t_1}$.

\subsubsection{Update $r^k$ for robust dispatch~\eqref{robust_mpc_conv}}
When there are disruptive events and the predicted requests number is a range $\hat{r}(h_1) \in [\hat{R}_1(h_1), \hat{R}_2(h_1)]$, similarly we set the uncertain set of $r^k$ as the following interval for the computationally efficient form of robust dispatch problem~\eqref{robust_mpc_conv}%an uncertain set of $\frac{1}{H}[\hat{R}_1(h_1), \hat{R}_2(h_1)]$. 
\begin{align*}
r^k \in \begin{cases}&\frac{1}{H} \left [\hat{R}_1(h_1)-r_o(h_1), \hat{R}_2(h_1)-r_o(h_1)\right], \\
 &\text{if}\ (k+h_2-1)t_2\leq h_1t_1, \\
                                         & \frac{1}{H}\left[\hat{R}_1(\ceil*{\frac{(k+h_2-1)t_2}{t_1}}), \hat{R}_2(\ceil*{\frac{(k+h_2-1)t_2}{t_1}})\right], \  \text{o.w.}\end{cases}
\end{align*}

\subsubsection{Spatial and temporal granularity of Algorithm~\ref{mpc_rt}}
The main computational cost of each iteration is on step (3), and $t_2$ should be no shorter than the computational time of the optimization problem. We regulate parameters according to experimental results based on a given data set, since there are no closed form equations to decide optimal design values of these parameters.

For the parameters we estimate from a given GPS dataset, the method we use in the experiments (but not restricted to it) will be discussed in Section~\ref{sec:simulation}. The length of every time slot depends on the predict precision of prediction, desired control outcome, and the available computational resources. We can set a large time horizon to consider future costs in the long run. However, in practice we do not have perfect predictions, thus a large time horizon may amplify the prediction error over time. 
Applying real-time information to adjust taxi supply is a remedy to this problem. Modeling techniques are beyond the scope of this work. If we have perfect knowledge of customer demand and taxi mobility models, we can set a large time horizon to consider future costs in the long run. However, in practice we do not have perfect predictions, thus a large time horizon may amplify the prediction error over time. Likewise, if we choose a small look-ahead horizon, then the dispatch solution may not count on idle distance cost of the future. Applying real-time information to adjust taxi supply is a remedy to this problem. 
With an approximated mobility pattern matrix $C^k$, the dispatch solution with large $T$ is even worse than small $T$.

\subsubsection{Selection process of parameters $\beta^k$, $\alpha^k$, and $T$}
The process of choosing values of parameters for Algorithm~\ref{mpc_rt} is a trial and adjusting  process, by increasing/decreasing the parameter value and observing the changing trend of the dispatch cost, till a desired performance is reached or some turning point occurs that the cost is not reduced any more. For instance, objective weight $\beta^k$ is related to the objective of the dispatch system, whether it is more important to reach fair service or reduce total idle distance. Some parameter is related to additional information available to the system besides real-time GPS and occupancy status data; for instance, $\alpha^k$ can be adjusted according to the average speed of vehicles or traffic conditions during time $k$ as discussed in subsection~\ref{alphak}. Adjustments of parameters such as objective weight $\beta^k$, idle distance threshold $\alpha^k$, prediction horizon $T$ when considering the effects of model accuracy, control objectives are shown in Section~\ref{sec:simulation}. A formal parameter selection method is a direction for future work.

\subsection{Multi-level dispatch framework}
%\textbf{Real-time information of waiting requests}:

We do not restrict the data source of customer demand -- it can be either predicted results or customer reservation records. Some companies provide taxi service according to the current requests in the queue. For reservations received by the dispatch center ahead of time, the RHC framework in Algorithm~\ref{mpc_rt} is compatible with this type of demand information --- we then assign value of the waiting requests vector $r^k$, taxi mobility function $f^k$ in~\eqref{opt2} according to the reservations, and the solution is subject to customer bookings. 

For customer requests received in real-time, a multi-level dispatch framework is available based on Algorithm~\ref{mpc_rt}. The process is as follows: run Algorithm~\ref{mpc_rt} with predicted demand $r^k$, and send dispatch solutions to vacant taxis. When vacant taxis arrive at dispatched locations, the dispatch center updates real-time demand such as bookings that recently appear in the system, then the dispatch method based on current demand such as the algorithm designed by Lee \emph{et al.}~\cite{Lee_review} can be applied. By this multi-level dispatch framework, vacant taxis are pre-dispatched at a regional level according to predicted demand using the RHC framework, and then specific locations to pick up a passenger who just  booked a taxi is sent to a vacant taxi according to the shortest time path~\cite{Lee_review}, with the benefit of real-time traffic conditions.
\iffalse
\subsubsection{Generalization of Algorithm~\ref{mpc_rt}}
\textbf{Distributed RHC algorithm}:
Since the relaxed form of iteration step (3) is polynomial of the problem size $V_n$, a centralized framework works well for a certain range of variable numbers based on the computational capability of the monitoring system (see Section~\ref{sec:simulation} for more detail). When centralized computation is not efficient enough for a large-scale problem, we can design a distributed information collecting and iterative computing RHC algorithm. In general a distributed framework introduces a trade-off between dispatch cost and computational complexity, and is one direction of future work.
\fi
%--total cost will increase though iteration speed will be faster. A distributed framework
%\end{remark}
%\begin{remark}

\section{Case Study: Method Evaluation}
\label{sec:simulation}
%%%%%%%%%%%%%%%%%%%%
%%%%%%%%%%%%%%%%%%%%%
\begin{table*}[t]
\vspace{-8pt}
\centering
\begin{tabular}{|c|c|c|c|c|c|c|}
  \hline
 \multicolumn{4}{|c|} {Taxicab GPS Data set }                         & \multicolumn{3}{|c|}{Format} \\ \hline
  Collection Period& Number of Taxis & Data Size & Record Number & ID & Status & Direction \\ \hline
 05/17/08-06/10/08 & $500$ & $90MB$ & $1,000,000$ & Date and Time & Speed & GPS Coordinates \\ \hline
 \end{tabular}
 \iffalse
 \begin{tabular}{|c|c|c|c|c|c|c|}
 Disturbance event &  Date  &   Time\\ \hline
                                   &              &            \\ \hline
 \end{tabular}
 \fi
% \vspace{-8pt}
      \caption{San Francisco Data in the Evaluation Section. Giant baseball game in AT\&T park on May 31, 2008 is a disruptive event we use for evaluating the robust optimization formulation.}
     %$region $3$ during $5:00-6:00$pm .}
     \label{datasf}
%     \vspace{-8pt}
\end{table*}%\vspace{-8pt}
\begin{figure*}[!t]
	\centering
	\vspace{-15pt}
	\subfigure[Requests during weekdays] %[Estimated requests of different hours of a day during weekdays] 
	{	
	%\begin{minipage}{.5\textwidth}
	%\centering
	\includegraphics [width=0.36\textwidth]{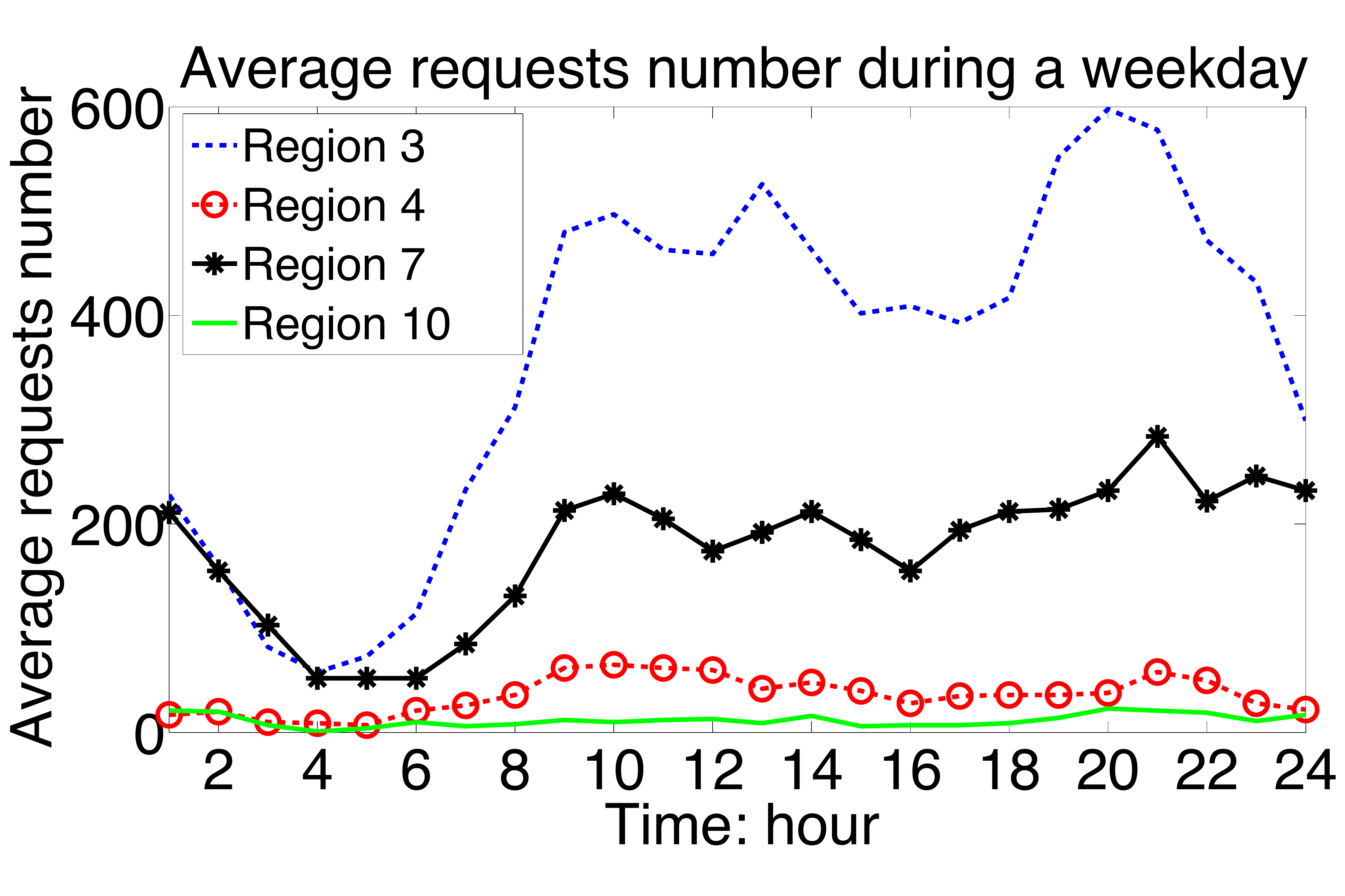} %40		35
		%	\caption{Estimated requests of different hours of a day during weekdays}
			\label{fig:r_weekday}
		}
\subfigure[Requests during weekends] %[Estimated requests of different hours of a day during weekends]
{
				\includegraphics[width=0.39\textwidth]{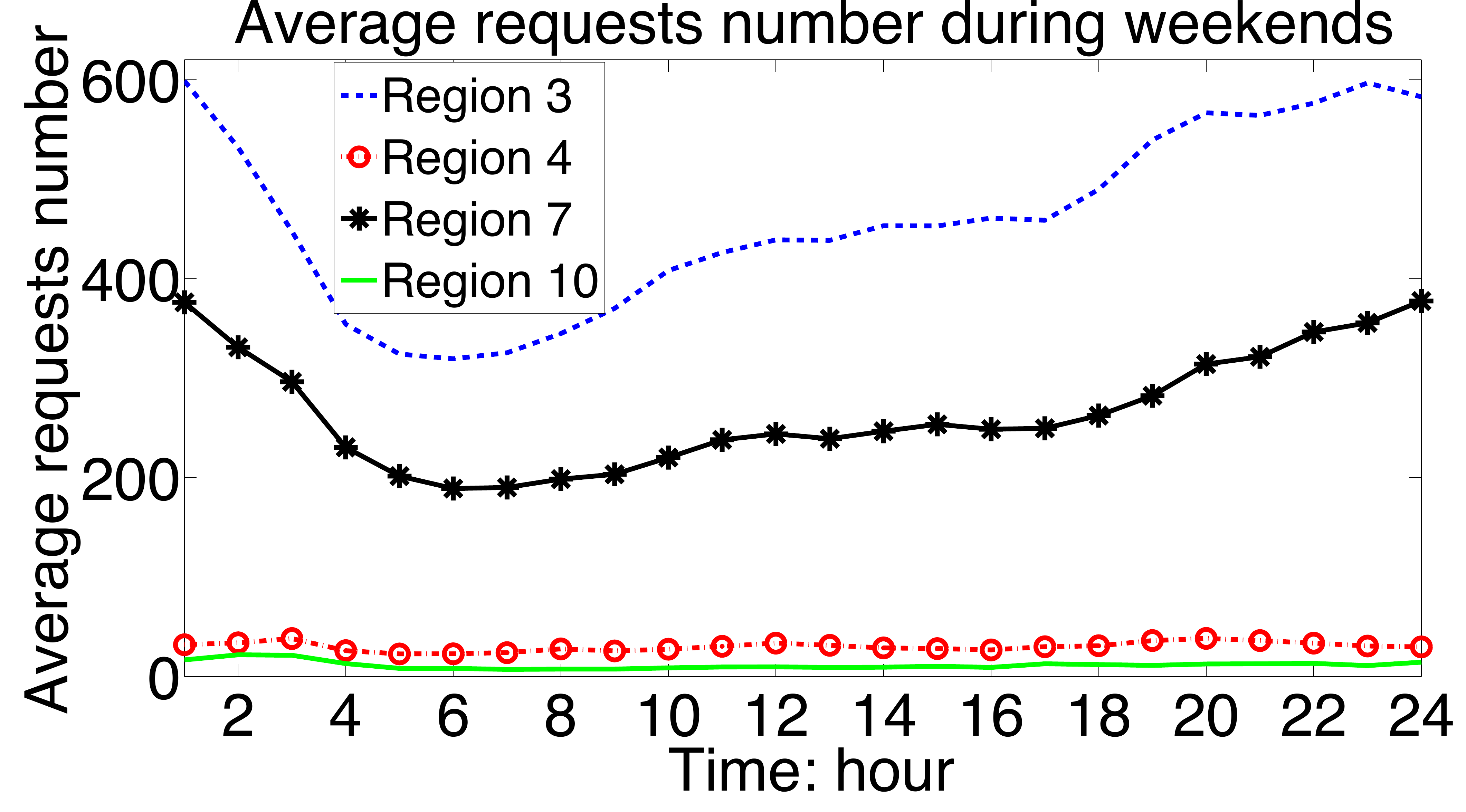}
				%\caption{Estimated requests of different hours of a day during weekends}
				\label{fig:r_weekend}
	}
	\subfigure [Drop off during weekdays]%[Estimated drop off events of different hours of a day during weekdays] 
	{	
			\includegraphics [width=0.35\textwidth]{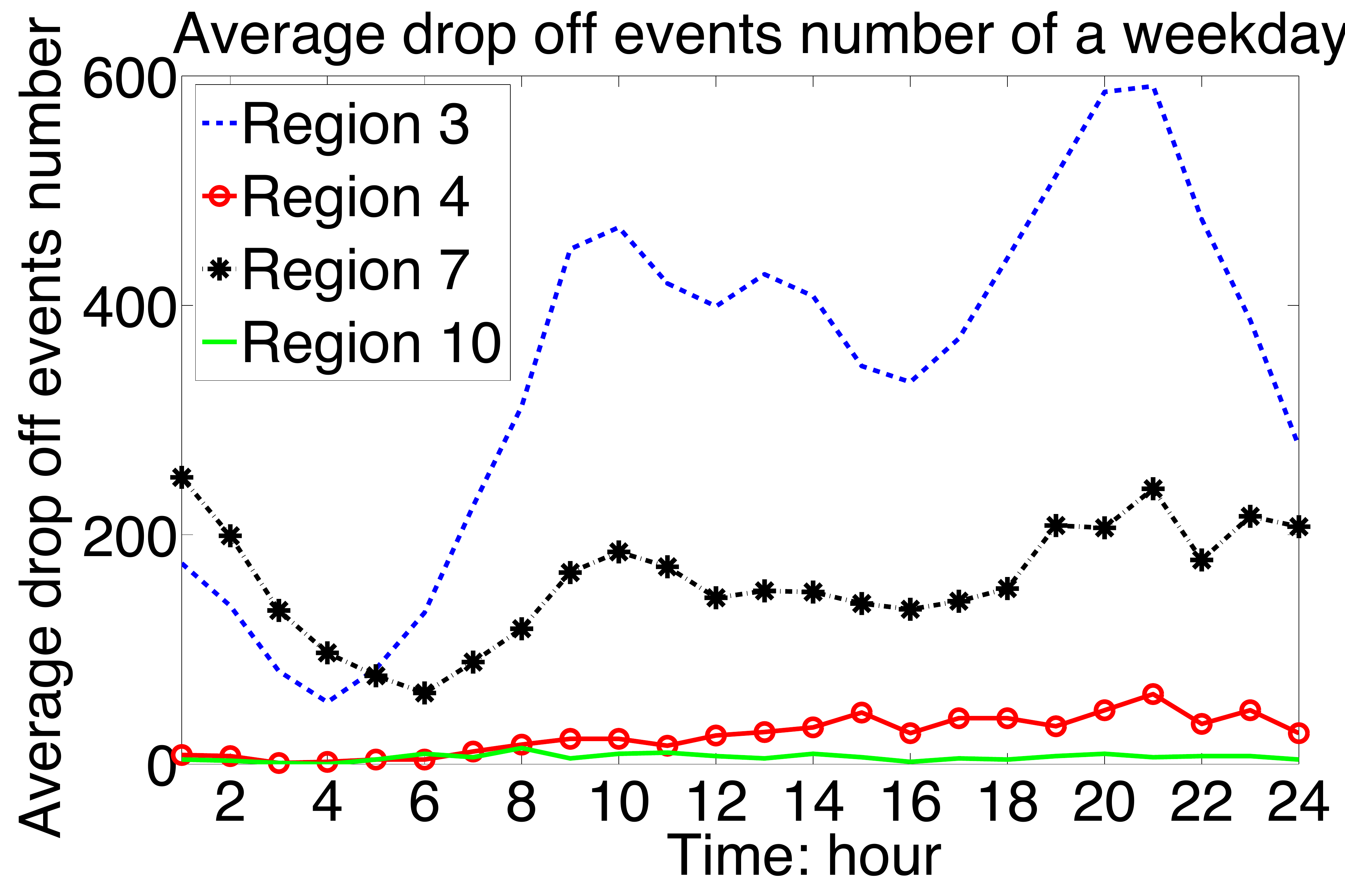} %40		35
			\label{fig:dp_weekday}
		}
\subfigure[Drop off during weekends] %[Estimated drop off events of different hours of a day during weekends]
{
				\includegraphics[width=0.37\textwidth]{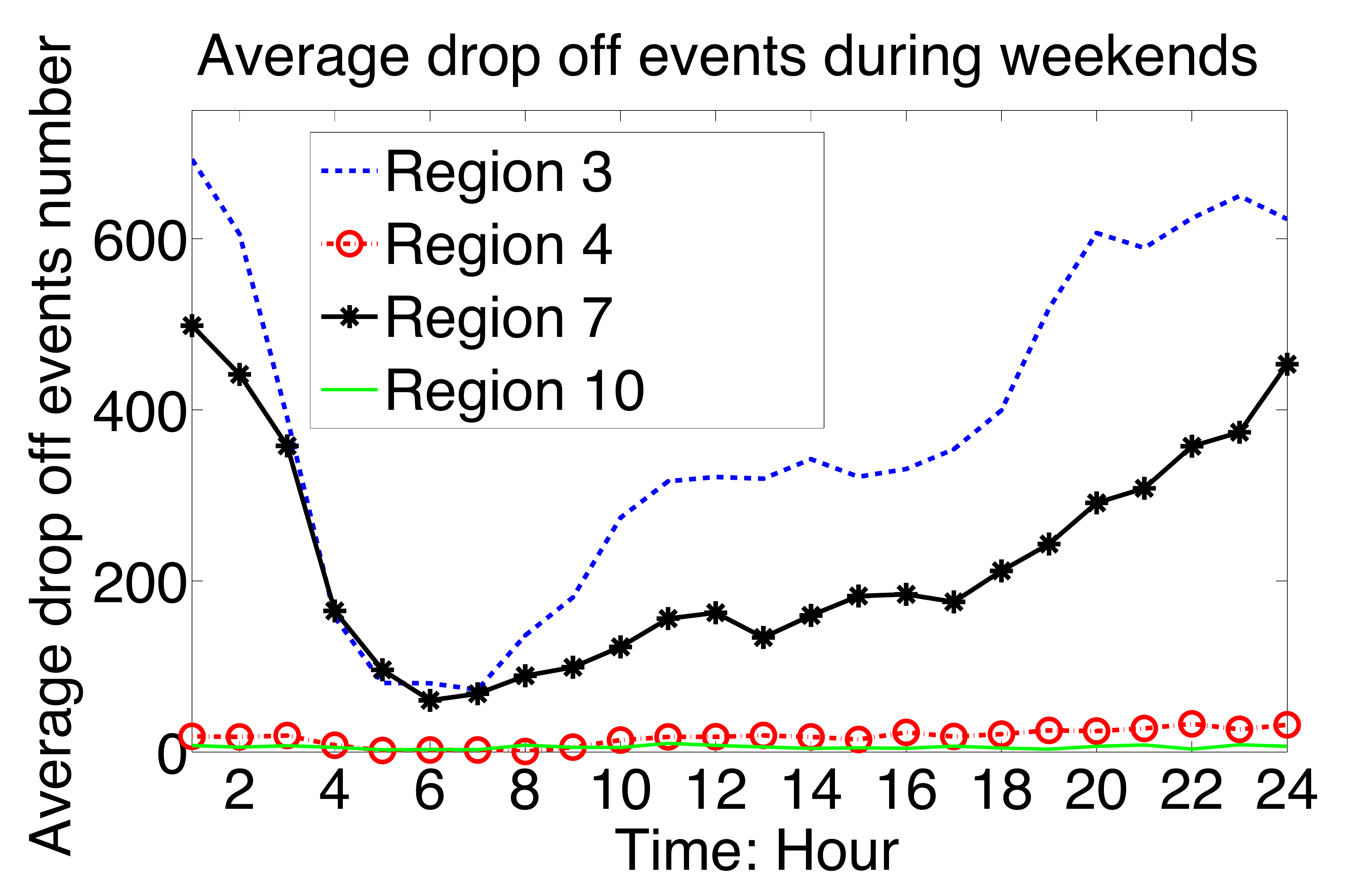}
				\label{fig:dp_weekend}
		}
\vspace{-10pt}
	\caption{Requests at different hours during weekdays and weekends, for four selected regions. A given historical data set provides basic spatiotemporal information about customer demands, which we utilize with real-time data to dispatch taxis. }%, and the results are part of given parameters for computing MPC algorithm.}
\label{fig:dp}
%\vspace{-8pt}
\end{figure*}

We conduct trace-driven simulations based on a San Francisco taxi data set~\cite{sf_data} summarized in Table~\ref{datasf}. In this data set, a record for each individual taxi includes four entries: the geometric position (latitude and longitude), a binary indication of whether the taxi is vacant or with passengers, and the Unix epoch time. With these records, we learn the average requests and mobility patterns of taxis, which serve as the input of Algorithm~\ref{mpc_rt}. We note that our learning model is not restricted to the data set used in this simulation, and other models~\cite{Dmodel}  and date sets can also be incorporated.

We implement Algorithm~\ref{mpc_rt} in Matlab using the optimization toolbox called CVX~\cite{cvx}. 
 We assume that all vacant taxis follow the dispatch solution and go to suggested regions. Inside a target region, we assume that a vacant taxi automatically picks up the nearest request recorded by the trace data, and we calculate the total idle mileage including distance across regions and inside a region by simulation. The trace data records the change of GPS locations of a taxi in a relatively small time granularity such as every minute. Moreover, there is no additional information about traffic conditions or the exact path between two consecutive data points when they were recorded. Hence, we consider the path of each taxi as connected road segments determined by each two consecutive points of the trace data we use in this section. Assume the latitude and longitude values of two consecutive points in the trace data are $[l_{x1}, l_{y1}]$ and $[l_{x2},l_{y2}]$, for a short road segment, the mileage distance between the two points (measured in one minute) is approximated as being proportional to the value $(|l_{x1}-l_{x2}|+|l_{y1}-l_{y2}|)$. The geometric location of a taxi is directly provided by GPS data. Hence, we calculate geographic distance directly from the data first, and then convert the result to mileage. 

Experimental figures shown in Subsection~\ref{mpc_eval} and~\ref{parameter} are average results of all weekday data from the data set~\ref{datasf}. Results shown in Subsection~\ref{robust_evaluate} are based on  weekend data.
\iffalse and represented as a vector $[l_x, l_y]$,  the latitude and longitude values of a point. We update GPS data in short period, and the distance a taxi drives during two updates is approximated by the geographical distance between point $[l_{x1}, l_{y1}]$ and point $[l_{x2},l_{y2}]$ -- norm $\sqrt{(l_{x1}-l_{x2})^2+(l_{y1}-l_{y2})^2}$. Then the total distance a taxi drives in a long time is the sum of small line segment geographical distance in all the short time periods.
\fi
%The simulation process  estimate inter region

%Figures and tables show average simulation results based on this data set.
 
%%%%%%%%%%%%%%%%%
%%%%%%%%%%%%%%%

\subsection{Predicted demand based on historical data} 
%We assume that customer demand is close to the number of pick up events,
Requests during different times of a day in different regions vary a lot, and Figure~\ref{fig:r_weekday} and Figure~\ref{fig:r_weekend} compare bootstrap results of requests $\hat{r}(h_1)$ on weekdays and weekends for selected regions. This shows a motivation of this work--- necessary to dispatch the number of vacant taxis according to the demand from the perspective of system-level optimal performance. The detailed process is described as follows.

The original SF data set does not provide the number of pick up events, 
hence %we need to process the original data to extract pick up and drop off information first. 
one intuitive way to determine a pick up (drop off) event is as follows. When the occupancy binary turns from $0$ to $1$ ($1$ to $0$), it means a pick up (drop off) event has happened. Then we use the corresponding geographical position to determine which region this pick up (drop off)
belongs to; use the time stamp data to decide during which time slot this pick up (drop off) happened. 
%Then add one to the pick up events counter of the corresponding region  time slot. Finally we get a $(h\times d)$ row, $n$ column requests matrix $r$, where $h$ is the number of time slots in one day, $d$ is the number of days, $n$ is the number of regions. 
After counting the total number of pick up and drop off events during each time slot at every region, we obtain a set of vectors $r_{d'}(h_k), dp_{d'}(h_k), d'=1,\dots, d$, where $d$ is the number of days for historical data . In the following analysis, each time slot $h_1$ is the time slot of predicting demand model chosen by the RHC framework. The SF data set includes about $24$ days of data, so we use $d=18$ for weekdays, and $d=6$ for weekends. The bootstrap process for a given sample time number $B=1000$ is given as follows.
%($d$ is about $24$ for SF data set)

%and we divide San Francisco to $n=16$ regions geometrically. 
%For every time slot $h_k$, we want to examine the average requests $r(h_k)\in \mathbb{R}^{1\times n}$ at every region, and estimate the variance and standard error of this average requests. Everyday we have a sample of $r(h_k)$, thus we have a single sample of values $r_1(h_k), \dots, r_d(h_k)$ available. The bootstrap procedure is then: 
(a) Randomly sample a size $d$ dataset with replacement from the data set $\{r_1(h_1), \dots, r_d(h_1)\}$, calculate 
$$\hat{r}^1(h_1)=\frac{1}{d} \sum_{d'=1}^{d} r_{d'}(h_1),\  \text{for}\  h_1=1,\dots, 1440/h_1.$$

(b) Repeat step (a) for $(B-1)$ times, so that we have $B$ estimates for each $h_1$,
\\\centerline{$
\hat{r}^b(h_1),\quad b=1,\dots, B.
$}
The estimated mean value of $\hat{r}(h_1)$ based on $B$ samples is
\\\centerline{$ 
%\begin{align}
\hat{r}(h_1)=\frac{1}{B}\sum_{l=1}^{B}\hat{r}^l(h_1).
%\label{average_r}
%\end{align}
$}
 
(c)  Calculate the sample variance of the B estimates of $r(h_1)$ for each $h_1$,
\begin{align}
\hat{V}_{\hat{r}(h_1)}=\frac{1}{B}\sum_{b=1}^{B}(\hat{r}^b(h_1)-\frac{1}{B}\sum_{l=1}^{B}\hat{r}^l(h_1)).
\label{variance_r}
\end{align}
%In the next section, we will see an example result of bootstrap for one time slot.
%$\mathbf{r}_k$: count the total pick up event at every region during the time slot $k$;

To estimate the demand range for robust dispatch problem~\eqref{robust_mpc_conv} according to historical data, we construct the uncertain set of demand $r^k$ based on the mean and variance of the bootstrapped demand model. For every region $j$, the boundary of demand interval is defined as
\begin{align}
\begin{split}
\tilde{R}_{1,j}(h_1)=\hat{r}_j(h_1)-\sqrt{\hat{V}_{\hat{r}(h_1),j}},\\
 \tilde{R}_{2,j}(h_1)=\hat{r}_j(h_1)+\sqrt{\hat{V}_{\hat{r}(h_1),j}},
 \end{split}
\end{align}
where $\hat{r}_j(h_1)$ is the average value of each step (b) and $\hat{V}_{\hat{r}(h_1),j}$ is the variance of estimated request number defined in~\eqref{variance_r}.
This one standard deviation range is used for evaluating the performance of robust dispatch framework in this work.
%defined as~\eqref{average_r},

Estimated drop off events vectors $dp(h_1)$ are also calculated via a similar process. Figure~\ref{fig:dp_weekday} and~\ref{fig:dp_weekend} show bootstrap results of passenger drop off events $dp(h_1)$ on weekdays and weekends for selected regions.

%state transition matrix of $5:00$pm.

%Then divide each element of estimated transition matrix $\hat{T}(h_2)$ by the sum of the row, we get 
 
%as the taxi transition matrix $\mathbf{T}^k$ for the time slot $k$. 
%Thus the bootstrap estimated $\hat{\mathbf{C}}(h_k)$ is calculated by equation~\eqref{C}. 
%
%\footnotesize
%\normalsize
\begin{table*}[t]
\centering
\begin{tabular}{|c|c|c|c|c|c|c|c|c|c|}
  \hline
 Region ID & 1 & 2 & 3 & 4 & 5 & 6 & 7 & 8\\ \hline
  Transit probability &    0.0032 &   0.0337  &  0.5144 &   0.0278 &   0.0132  &  0.0577  &0.1966  &  0.0263    \\  \hline
Region ID  & 9 & 10 & 11 & 12 & 13 & 14 & 15 & 16\\ \hline
Transit probability &      0.0001  &  0.0050  &  0.0340  &  0.0136  &  0.0018  &   0.0082  &  0.0248  &  0.0396 \\ \hline  
%Region ID & &\\ \hline
%Transit probability& &   \\ \hline 
 \end{tabular}
% \vspace{-8pt}
     \caption{An estimation of state transition matrix by bootstrap: one row of matrix $\hat{C}(h_k)$}
     %$region $3$ during $5:00-6:00$pm .}
     \label{boot_c}
     \vspace{-15pt}
\end{table*}

%\iffalse
%\textbf{Generate allocating coordinates for vacant taxis}:
%Given the region partition function $L$, the city is divided into $n$ regions geographically.
 %each region is assigned a two dimensional center geographical position $W_{j0}, j \in \{1,\dots, n\}$.
For evaluation convenience, we partition the city map to regions with equal area. To get the longitude and latitude position $W_{i}\in \mathbb{R}^{n\times 2}$ of vacant taxi $i$, we randomly pick up a station position in the city drawn from the uniform distribution.

%With an estimation of passenger requests and future service ability of occupied taxis, we apply it as a predicted model for the MPC dispatch algorithm. 
\subsection{RHC with real-time sensor information}
\label{mpc_eval}
\begin{figure}[b!]
\vspace{-10pt}
\centering
\includegraphics [width=0.39\textwidth]{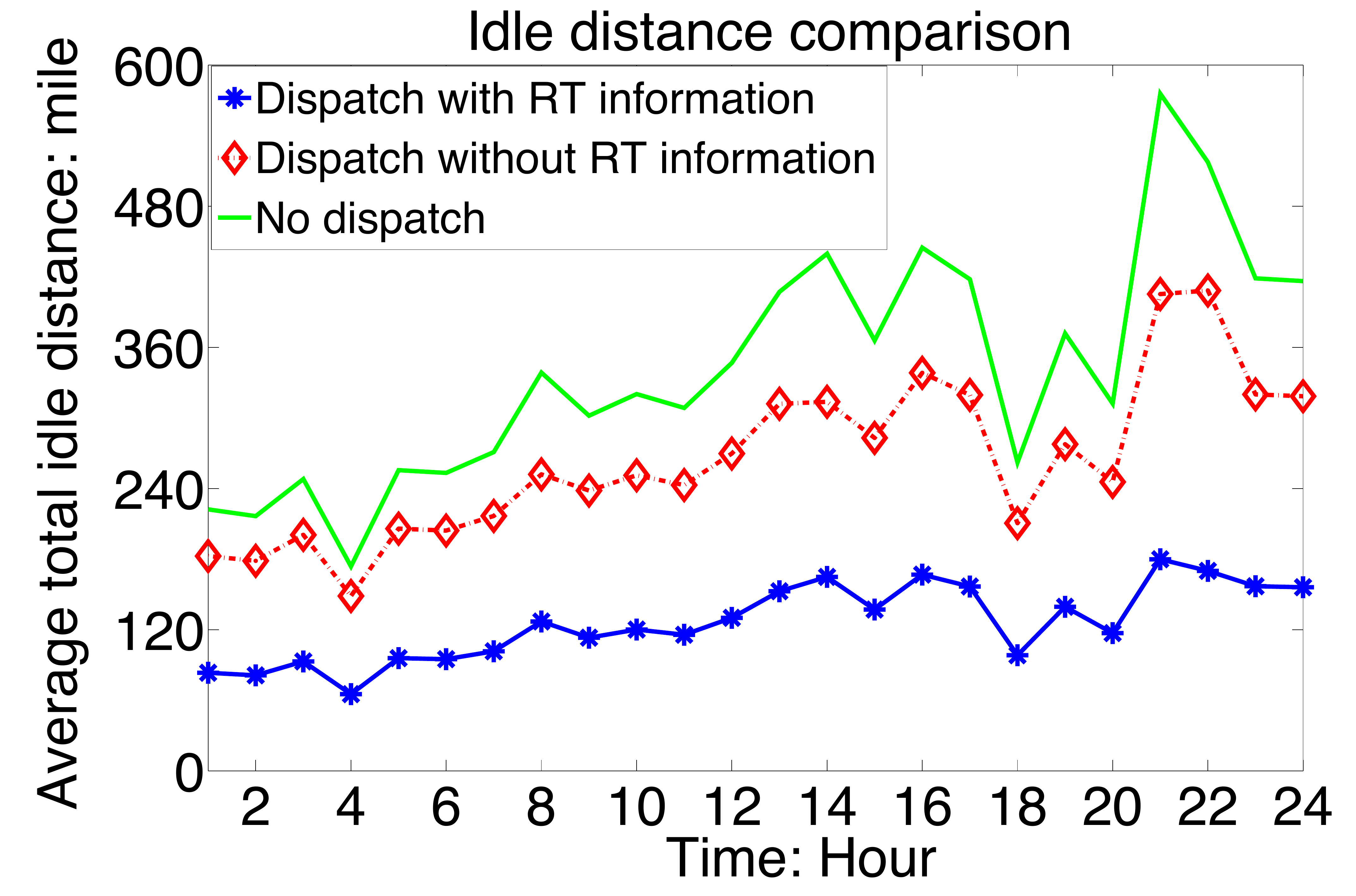}
\vspace{-10pt}
\caption{Average idle distance comparison for no dispatch, dispatch without real-time data, and dispatch with real-time GPS and occupancy information. Idle distance is reduced by $52\%$ given real-time information, compared with historical data without dispatch solutions.}
\label{compare_gps}
%\vspace{-5pt}
\end{figure}

To estimate a mobility pattern matrix $\hat{C}(h_2)$, we define a matrix $T(h_2)$, where $T(h_2)_{ij}$ is the total number of passenger trajectories that starting at region $i$ and ending at region $j$ during time slot $h_2$. We also apply bootstrap process to get $\hat{T}(h_2)$, and $\hat{C}(h_2)_{ij}=\hat{T}(h_2)_{ij}/(\sum\limits_{j}\hat{T}(h_2)_{ij}).$ Table~\ref{boot_c} shows one row of $\hat{C}(h_{2})$ for 5:00-6:00 pm during weekdays, the transition probability for different regions. The average cross validation error for estimated mobility matrix $\hat{C}(h_{2})$ of time slot $h_2$, $h_2=1,\dots, 24$ during weekdays is $34.8\%$, which is a reasonable error for estimating total idle distance in the RHC framework when real-time GPS and occupancy status data is available. With only estimated mobility pattern matrix $\hat{C}(h_{2})$, the total idle distance is reduced by $17.6\%$ compared with the original record without a dispatch method, as shown in Figure~\ref{compare_gps}. We also tested the case when the dispatch algorithm is provided with the true mobility pattern matrix $C^k$, which is impossible in practice, and the dispatch solution reduces the total idle distance by $68\%$ compared with the original record. When we only have estimated mobility pattern matrices instead of the true value to determine ending locations and potential total idle distance for solving problem~\eqref{opt2} or~\eqref{robust_mpc_conv}, updating real-time sensing data compensates the mobility pattern error and improves the performance of the dispatch framework.
%We can see that this row sums up to one since the matrix is a state transition matrix. 

Real-time GPS and occupancy data provides latest position information of all vacant and occupied taxis. When dispatching available taxis with true initial positions, the total idle distance is reduced by $52\%$ compared with the result without dispatch methods, as shown in Figure~\ref{compare_gps}, which is compatible with the performance when both true mobility pattern matrix $C^k$ and real-time sensing data are available. This is because the optimization problem~\eqref{opt2} returns a solution with smaller idle distance cost given the true initial position information $P^0$, instead of estimated initial position provided only by mobility pattern of the previous time slot in the RHC framework. Figure~\ref{compare_gps} also shows that even applying dispatch solution calculated without real-time information is better than non dispatched result. 

%The taxi supply matches customer demand better while average idle mileage is reduced under this dispatch framework.

%Without real-time sensor information, positions of vacant taxis are estimated or predicted, and not accurate enough compared with real-time GPS information, thus introducing larger idle mileage
Based on the partition of Figure~\ref{sf_partition}, 
Figure~\ref{compare_ratio} shows that the supply demand ratio at each region of the dispatch solution with real-time information is closest to the supply demand ratio of the whole city, and the error $\left\|\frac{1}{N}\mathbf{1}^T_N X^k - \frac{1}{R^k}r^k\right\|_1$ is reduced by $45\%$ compared with no dispatch results. 
\iffalse
As explained in Section~\ref{sec:prob_form}, the objective to minimize $\left\|\frac{\mathbf{1}^T_N X^k}{N} - \frac{r^k}{\mathbf{1}^T_n r^k }\right\|_1$ is equivalent to make the supply/demand ratio at each region closest to the supply/demand ratio of the whole city under given constraints. 
\fi
Even the supply demand ratio error of dispatching without real-time information is better than no dispatch solutions. We still allocate vacant taxis to reach a nearly balanced supply demand ratio regardless of their initial positions, but idle distance is increased without real-time data, as shown in Figure~\ref{compare_gps}.
%This is because we sacrifice idle geographical distance cost -- as shown in Figure~\ref{compare_gps}, total idle distance increases without real-time information -- . 
Based on the costs of two objectives shown in Figures~\ref{compare_gps} and~\ref{compare_ratio}, the total cost is higher without real-time information, mainly results from a higher idle distance.
%\begin{figure}[b!]

\iffalse
There are usually more requests in places of interests and downtown business areas, and experimental results shown in Figure~\ref{compare_d} indicates busy areas with higher request bars in Region $3$ and Region $7$.
%Figure~\ref{compare_d} compares the proportions of taxi supply and customer demand at different regions during one time slot, for non dispatched historical data, dispatching without real-time sensor information and dispatching with real-time GPS and occupancy data. 
The bars of predicted requests and taxis show the values 
$\frac{r^k_j}{R^k}$($[\frac{\mathbf{1}^T_N X_{\cdot j}^k}{N}]$)---the proportion of requests (vacant taxis)  in each region compared with the total number of requests (vacant taxis) in the whole city. With a dispatch solution, areas with more requests are provided more vacant taxis.
\fi

%The bars of dispatched taxis are closer to the requests bar. 
%The bars for vacant taxis show available supply ratio at each region divided by the supply number of the whole city. 
\begin{figure}[t!]
%\vspace{-10pt}
\centering
\includegraphics [width=0.30\textwidth]{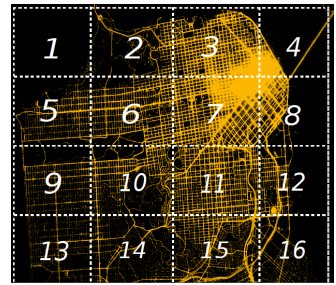}
\vspace{-10pt}
\caption{Heat map of passenger picking-up events in San Francisco (SF) with a region partition method. Region $3$ covers several busy areas, include Financial District, Chinatown, Fisherman Wharf. Region $7$ is mainly Mission District, Mission Bay, the downtown area of SF.} 
\label{sf_partition}
\vspace{-15pt}
\end{figure} 
\begin{figure}[b!]
\vspace{-15pt}
\centering
\includegraphics [width=0.39\textwidth]{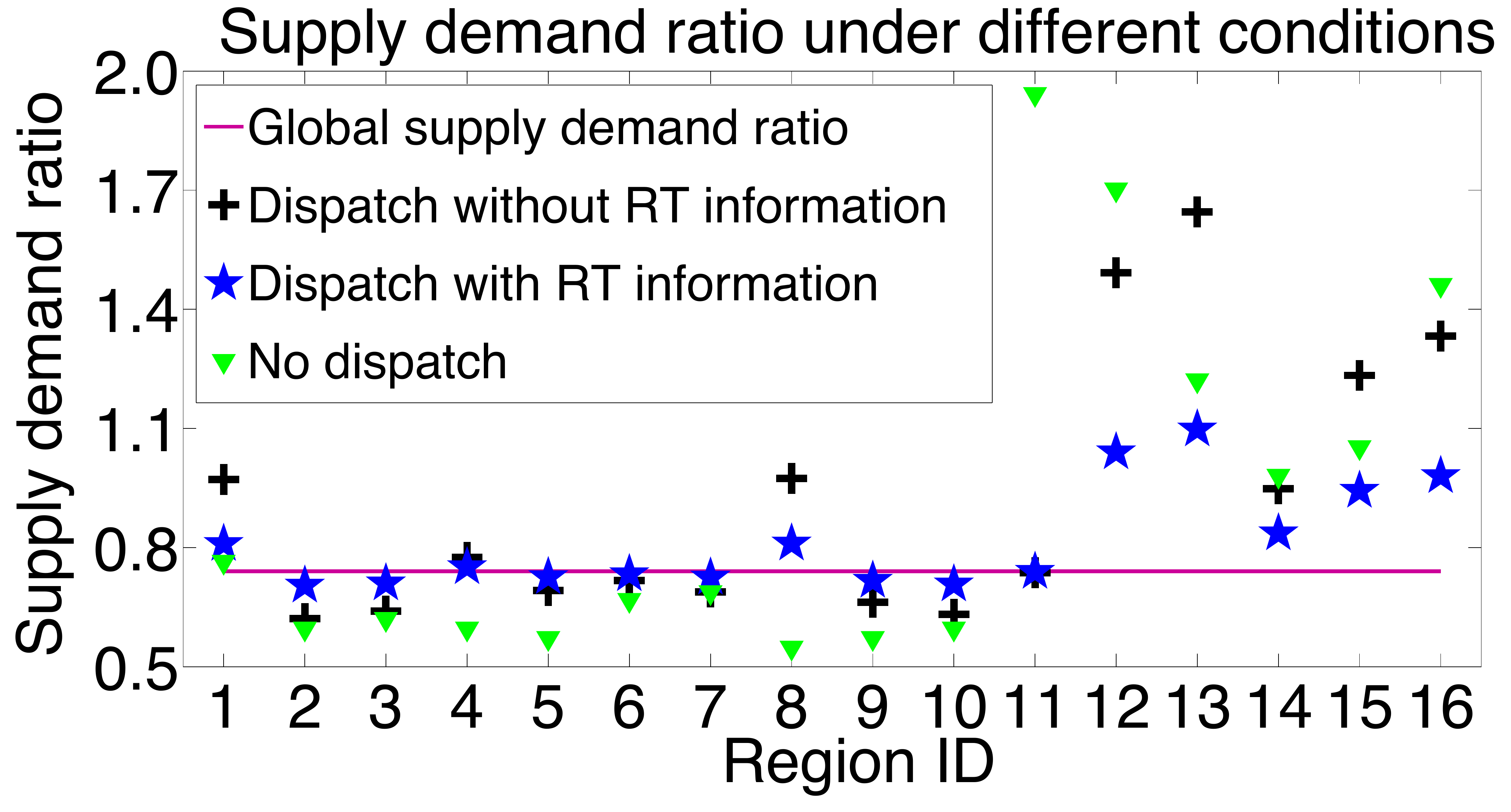}
\vspace{-10pt}
\caption{Supply demand ratio of the whole city and each region for different dispatch solutions. With real-time GPS and occupancy data, the supply demand ratio of each region is closest to the global level. The supply demand ratio mismatch error is reduced by $45\%$ with real-time information, compared with historical data without dispatch solutions.}
% distribution of vacant taxis is not balanced without dispatch method, since the supply/demand ratio among regions vary much.} 
\label{compare_ratio}
%\vspace{-10pt}
\end{figure} 
%Without dispatch the distribution of vacant taxis is not balanced, since the supply/demand ratio at some regions are far away from the global ratio. 
%Figure~\ref{compare_ratio} shows that with a dispatching solution, the supply/demand ratio at each region is closer to the steady state ratio -- the supply/demand ratio of the whole city, 
\iffalse
\begin{figure}[b!]
%\vspace{-8pt}
\centering
\includegraphics [width=0.36\textwidth]{compare_dn.pdf}
\vspace{-5pt}
\caption{Proportion of vacant taxis at selected regions under three conditions:  historical data of taxi positions without dispatch, dispatch without real-time data, and dispatch with real-time information. %Compared with total vacant taxis in the whole city for one time slot,  
The bar of dispatch solutions with real-time information matches the request bar best in most regions.} 
\label{compare_d}
\vspace{-8pt}
\end{figure} 
%\textbf{Estimate region transition matrix $\hat{C}(h_2)$}: 
\fi

\iffalse
at each region, we randomly pick a point according to a uniform distribution symmetric around the approximated center $W_{j0}$ for region $j$, and generate $W_{ij} \in \mathbb{R}^{1 \times 2}$. 
\fi
%\fi
%%%%%%%%%%%%%
%%%%%%%%%%%%%
%\label{robustmpc}
\begin{figure}[b!]
\vspace{-15pt}
\centering
\includegraphics [width=0.37\textwidth]{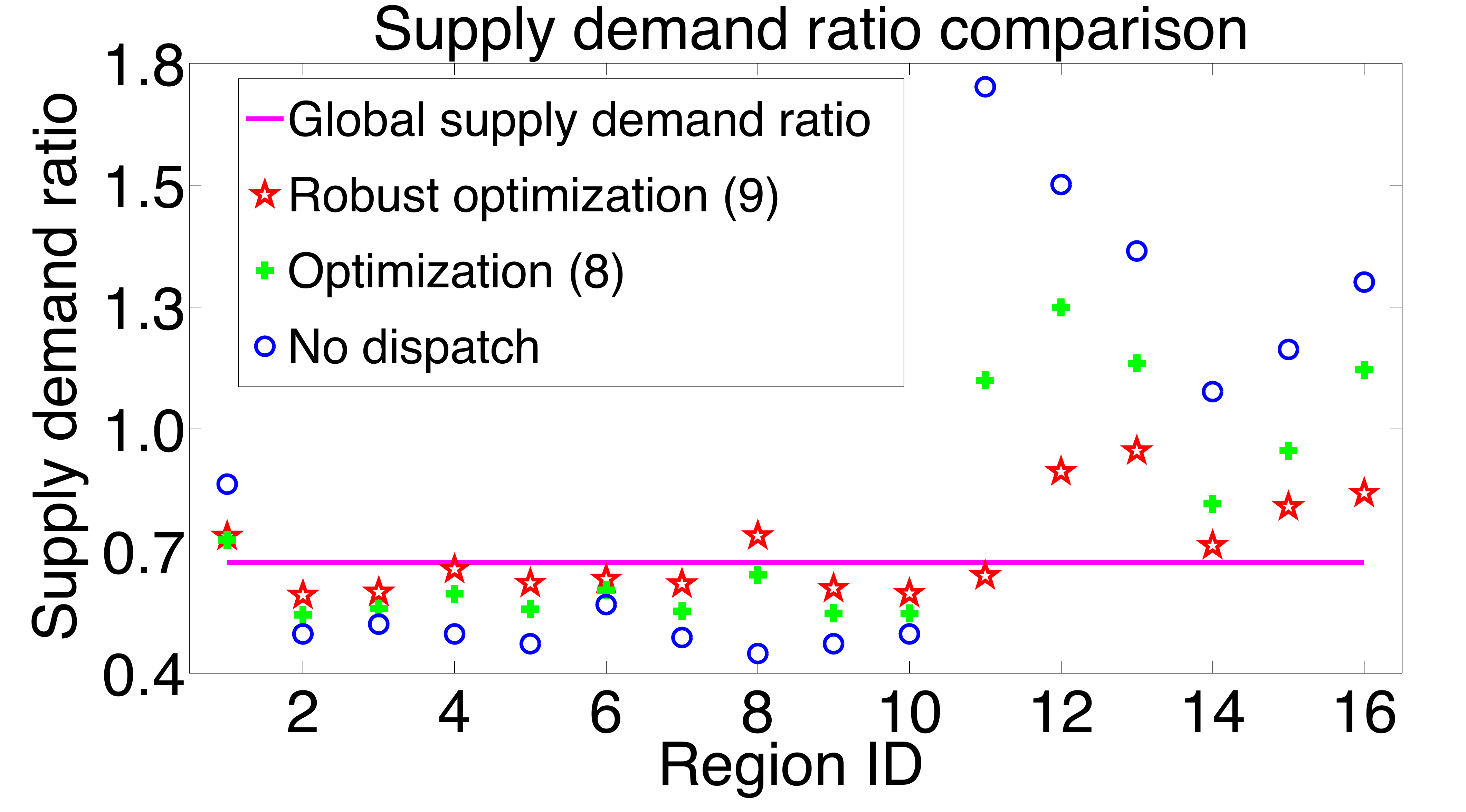}
\vspace{-10pt}
\caption{Comparison of supply demand ratio at each region under disruptive events, for solutions of robust optimization problems~\eqref{robust_mpc}, problem~\eqref{opt2} in the RHC framework, and historical data without dispatch. With the roust dispatch solutions of~\eqref{robust_mpc}, the supply demand ratio mismatch error is reduced by $46\%$.}
\label{fig:robust}
\end{figure}  %the supply demand ratio of each region is closer to the ratio of the whole city.
%%%%%%%%%%%%%%%%%%%%%%%%%%%%%%%%%%%%%%%
\iffalse
\begin{figure}[b!]
\centering
\includegraphics [width=0.39\textwidth]{compare_robustd.pdf}
\vspace{-10pt}
\caption{Average idle distance comparison under a disruptive event --- baseball game with solutions of~\eqref{opt2},~\eqref{robust_mpc} and ~\eqref{robust2}. Idle distance is reduced $52\%$ given real-time information, compared with historical data without dispatch solutions.}
\label{compare_robustd}
\vspace{-5pt}
\end{figure} 
\fi
%%%%%%%%%%%%%%%%%%%%%%%%%%%%%%%%%%%%%%%%
\subsection{Robust taxi dispatch}
\label{robust_evaluate}
\iffalse
One advantage of the RHC framework proposed in this paper is its compatibility with different problem formulations in iteration step (3). \fi

One disruptive event of the San Francisco data set is Giant baseball at AT\&T park, and we choose the historical record on May 31, 2008 as an example to evaluate the robust optimization formulation~\eqref{robust_mpc}. Customer request number for areas near AT\&T park is affected,  %-- Region $3, 4, 7$ are affected, 
especially Region $7$ around the ending time of the game, which increases about $40\%$ than average value. 

%and compare the the supply demand ratio error $\left\|\frac{1}{N}\mathbf{1}^T_N X^k - \frac{1}{R^k}r^k\right\|_1$ --- for a piecewise inequality form of uncertain set we solve problem~\eqref{robust_mpc}. 

Figure~\ref{fig:robust} shows that with dispatch solution of the robust optimization formulation~\eqref{robust_mpc},  the supply demand mismatch error $\left\|\frac{1}{N}\mathbf{1}^T_N X^k - \frac{1}{R^k}r^k\right\|_1$ is reduced by $25\%$ compared with the solution of~\eqref{opt2} and by $46\%$ compared with historical data without dispatch. The performance of robust dispatch solutions does not vary significantly and depends on what type of predicted uncertain demand is available when selecting the formulation of robust dispatch method. Even under solutions of~\eqref{opt2}, the total supply demand ratio error is reduced $28\%$ compared historical data without dispatch.  In general, we consider the factor of disruptive events in a robust RHC iteration, thus the system level supply distribution responses to the demand better under disturbance.

\subsection{Design parameters for Algorithm~\ref{mpc_rt}}
\label{parameter}
Parameters like the length of time slots, the region division function, the objective weight parameter and the prediction horizon $T$ of Algorithm~\ref{mpc_rt} affect the results of dispatching cost in practice. Optimal values of parameters for each individual data set can be different. Given a data set, we change one parameter to a larger/smaller value while keep others the same, and compare results to choose a suboptimal value of the varying parameter. We compare the cost of choosing different parameters for Algorithm~\ref{mpc_rt}, and explain how to adjust parameters according to experimental results based on a given historical data set with both GPS and occupancy record.

%Note that the requests in all figures are estimated waiting requests, a given parameter for~\eqref{opt2}. 
%Geographical distance can be considered approximately proportional to mileage in one city, 
\textbf{How the objective weight of~\eqref{opt2} -- $\beta^k$ affects the cost:}
\begin{table}[t]
%\vspace{-8pt}
\centering
\begin{tabular}{|c|c|c|c|c|}
  \hline
 $\beta^k$ & 0 & 2 & 10  & without dispatch\\ \hline
 s/d error &0.645 &1.998  &2.049  &2.664 \\  \hline
 idle distance               &3.056 &1.718   &1.096  &4. 519  \\ \hline
total cost                &0.645   & 5.434   &13.009    & 47.854  \\ \hline  
 \end{tabular}
% \vspace{-8pt}
     \caption{Average cost comparison for different values of $\beta^k$. }
     \label{table_beta}
     \vspace{-20pt}
\end{table}
%\fi
%%%%%%%%%%%%%%change figure to supply/demand ratio
%%%%%%%%%%%%%%
%%%%%%%%%%%%%%%%%%%%
%%%%%%%%%%%%%%%%%%%%%
\begin{figure}[b!]
\vspace{-15pt}
\centering
\includegraphics [width=0.38\textwidth]{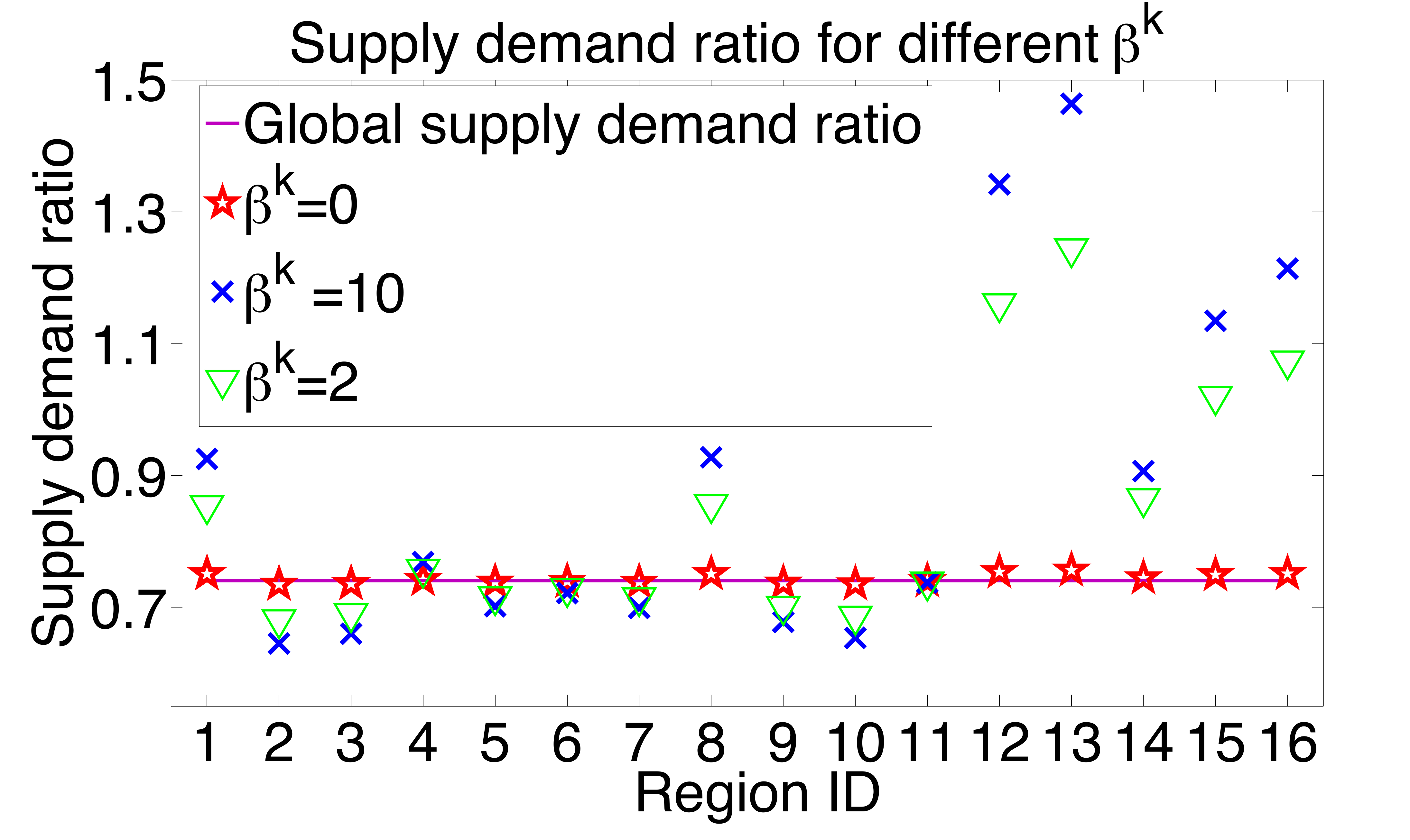}
\vspace{-10pt}
\caption{Comparison of supply demand ratios at each region during one time slot for different $\beta^k$ values. When $\beta^k$ is smaller, we put less cost weight on idle distance that taxis are allowed to run longer to some region, and taxi supply matches with the customer requests better.}
\label{compare_betaratio}
%\vspace{-10pt}
\end{figure} 
\begin{figure}[b!]
%\vspace{-8pt}
\centering
\includegraphics [width=0.37\textwidth]{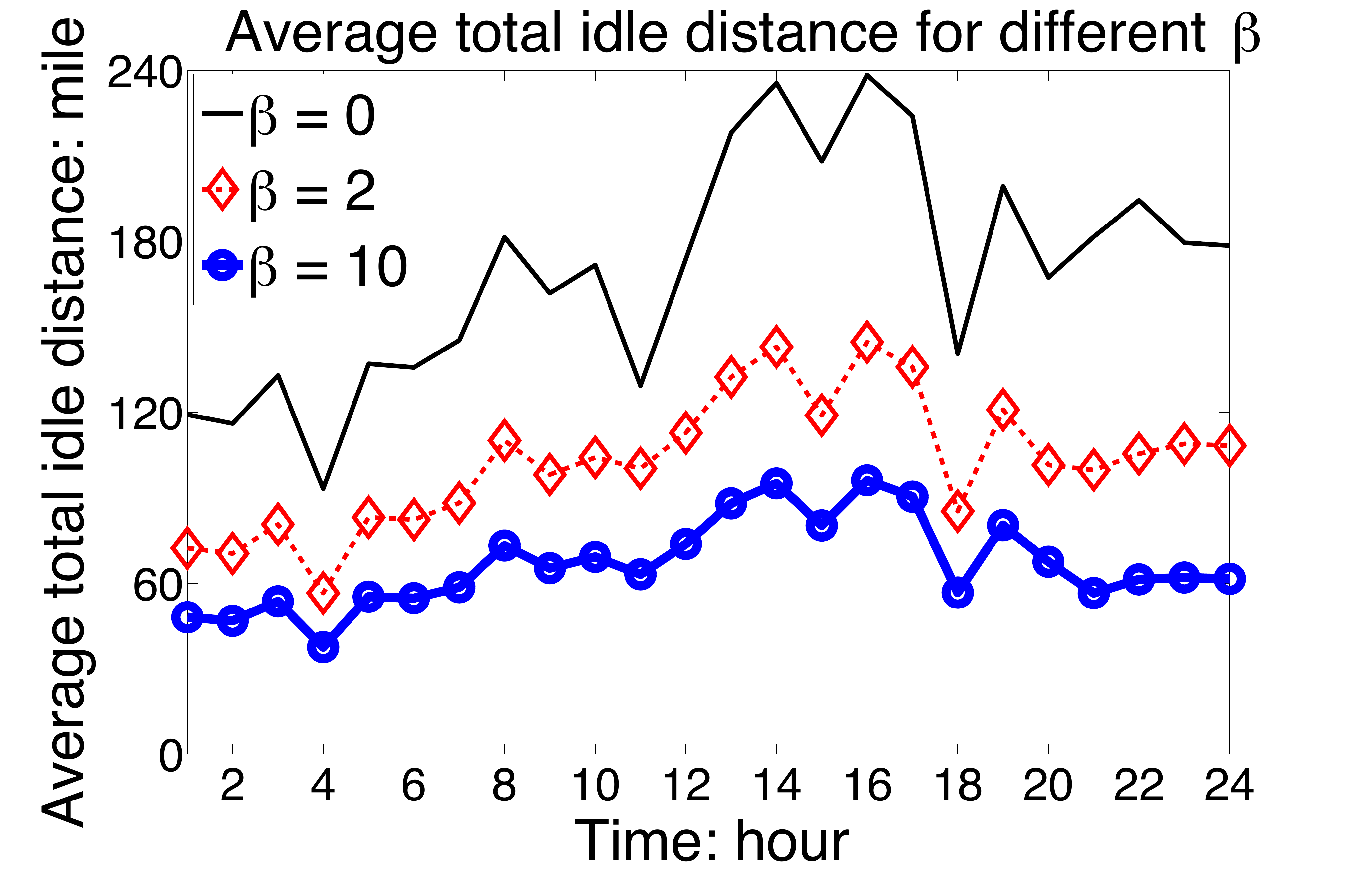}
\vspace{-10pt}
\caption{Average total idle distance of taxis at different hours. When $\beta^k$ is larger, the idle distance cost weights more in the total cost, and the dispatch solution causes less total idle distance.}
\label{compare_beta}
%\vspace{-8pt}
\end{figure} 
The cost function includes two parts --the idle geographical distance (mileage) cost and the supply demand ratio mismatch cost. This trade-off between two parts is addressed by $\beta^k$, and the weight of idle distance increases with $\beta^k$. A larger $\beta^k$ returns a solution with smaller total idle geographical distance, while a larger error between supply demand ratio, i.e., a larger $\left\|\frac{1}{N}\mathbf{1}^T_N X^k - \frac{1}{R^k}r^k\right\|_1$ value. 
The two components of the cost with different $\beta^k$ by Algorithm~\ref{mpc_rt}, and historical data without Algorithm~\ref{mpc_rt} are shown in Table~\ref{table_beta}. The supply demand ratio mismatch is shown in the s/d error column. 

We calculate the total cost as (s/d error $+\beta^k \times$ idle distance) (Use $\beta^k=10$ for the without dispatch column). Though with $\beta^k=0$ we can dispatch vacant taxis to make the supply demand ratio of each region closest to that of the whole city, a larger idle geographical distance cost is introduced compared with $\beta^k=2$ and $\beta=10$. %To conservatively estimate the decrease amount of total idle geographical distance, 
Compare the idle distance when $\beta^k=0$ with the data without dispatch, we get $23\%$ reduction; compare the supply demand ratio error of $\beta^k=10$ with the data without dispatch, we get $32\%$. 

Average total idle distance during different hours of one day for a larger $\beta^k$ is smaller, as shown in Figure~\ref{compare_beta}. The supply demand ratio error at different regions of one time slot is increased with larger $\beta^k$, as shown in Figure~\ref{compare_betaratio}. 

\textbf{How to set idle distance threshold $\alpha^k$:} Figure~\ref{compare_alphak} compares the error between local supply demand ratio and global supply demand ratio. Since we directly use geographical distance measured by the difference between longitude and latitude values of two points (GPS locations) on the map, the threshold value $\alpha^k$ is small --- $0.1$ difference in GPS data corresponds to almost $7$ miles distance on the ground. 
When $\alpha^k$ increase, the error between local supply demand ratio and global supply demand ratio decreases, since vacant taxis are more flexible to traverse further to meet demand. 

\begin{figure}[b!]
\vspace{-15pt}
\centering
\includegraphics [width=0.38\textwidth]{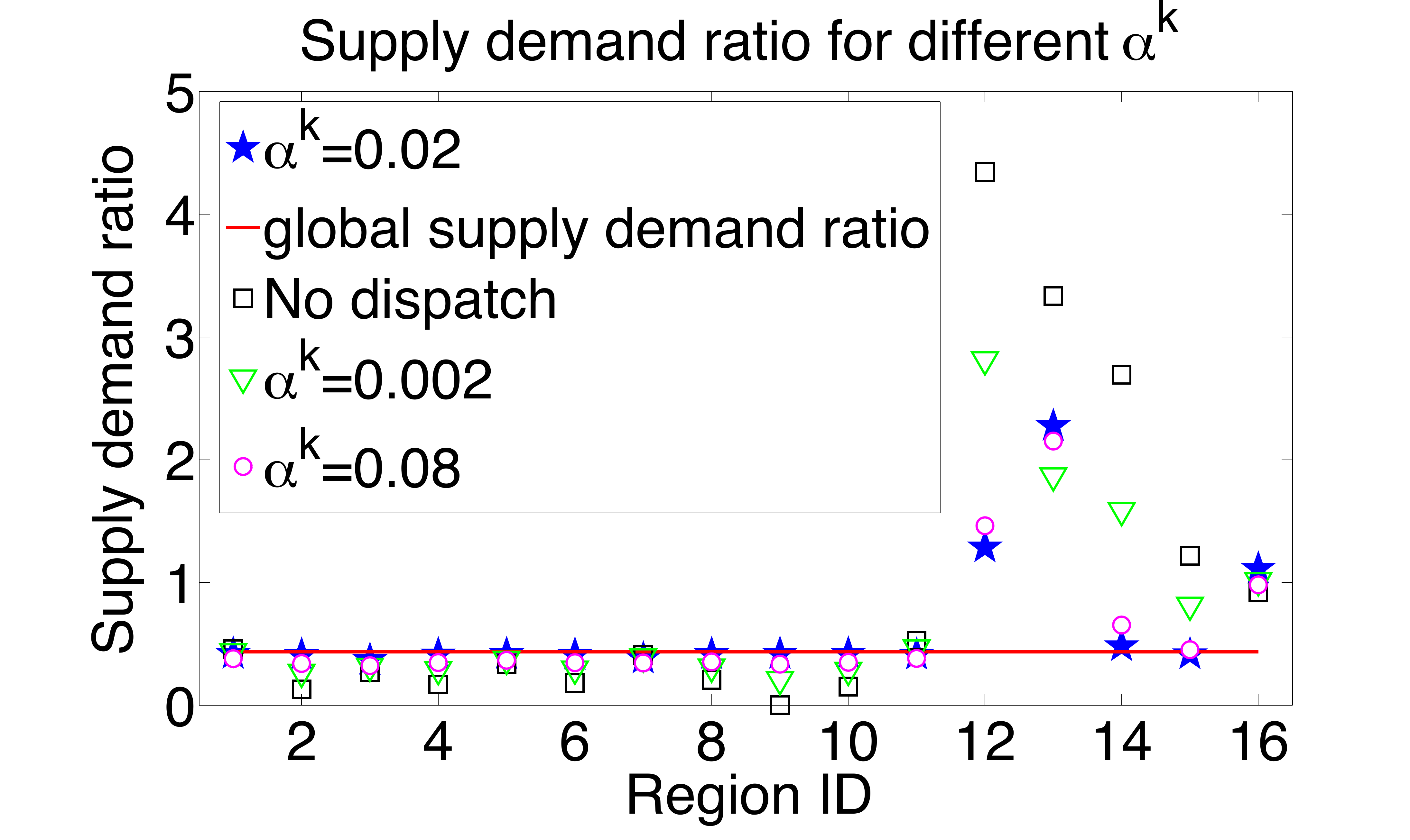}
\vspace{-10pt}
\caption{Comparison of supply demand ratios at each region during one time slot for different $\alpha^k$. When $\alpha^k$ is larger, vacant taxis can traverse longer to dispatched locations and match with customer requests better.}
\label{compare_alphak}
%\vspace{-10pt}
\end{figure}

\textbf{How to choose the number of regions}:
\begin{figure}[b!]
\vspace{-8pt}
\centering
\includegraphics [width=0.38\textwidth]{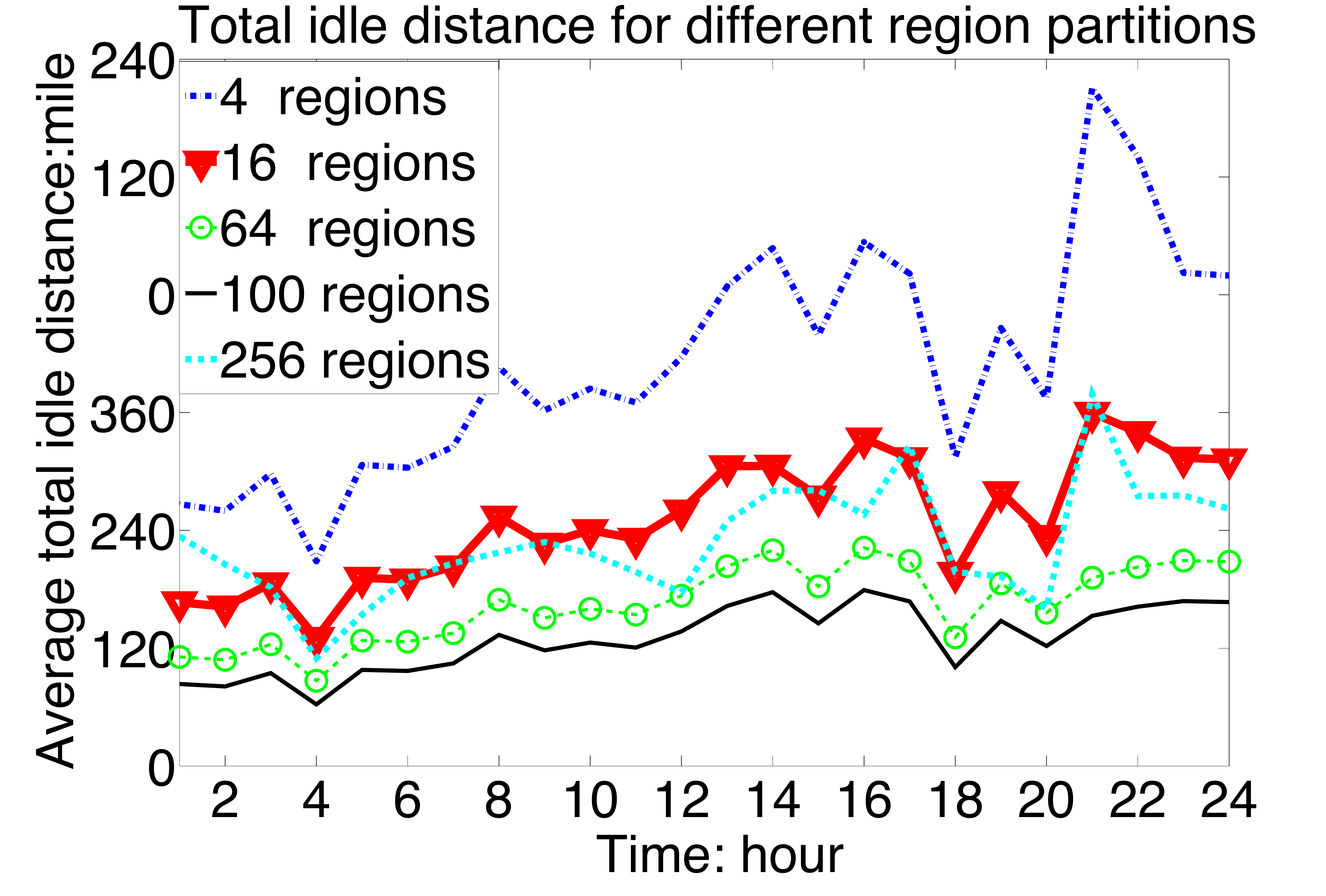}
\vspace{-10pt}
\caption{Average total idle distance of all taxis during one day, for different region partitions. Idle distance decreases with a larger region-division number, till the number increases to a certain level.}
\label{compare_region}
%\vspace{-8pt}
\end{figure} 
In general, the dispatch solution of problem~\eqref{opt2} for a vacant taxi is more accurate by dividing a city into regions of smaller area, since the dispatch is closer to road-segment level. However, we should consider other factors when deciding the number of regions, like the process of predicting requests vectors and mobility patterns based on historical data. A linear model we assume in this work is not a good prediction for future events when the region area is too small, since pick up and drop off events are more irregular in over partitioned regions.  
While Increasing $n$, we also increase the computation complexity. Note that the area of each region does not need to be the same as we divide the city in this experiment. 
\iffalse
since the number of variables in~\eqref{opt2} is linear with $n$ and the computational time to solve~\eqref{opt2} is polynomial of $n$. The the MPC framework is compatible with different region division methods according to design requirements of specific cities and dispatch systems.\fi
%Another reason that the region number is not the larger the better is  
%For areas with fewer requests and drop off events, accurate location dispatch of taxis may not necessary or possible, since the position of requests are randomly without a certain pattern. 
%To save computation cost in practice, we can divide busy areas to sub regions first, then consider to divide non-busy regions. 

Figure~\ref{compare_region} shows that the idle distance will decrease with a larger region division number, but the decreasing rate slows down; while the region number increases to a certain level, the idle distance almost keeps steady. 

\textbf{How to decide the prediction Horizon $T$}: 
\begin{figure}[t!]
%\vspace{-5pt}
\centering
\includegraphics [width=0.38\textwidth]{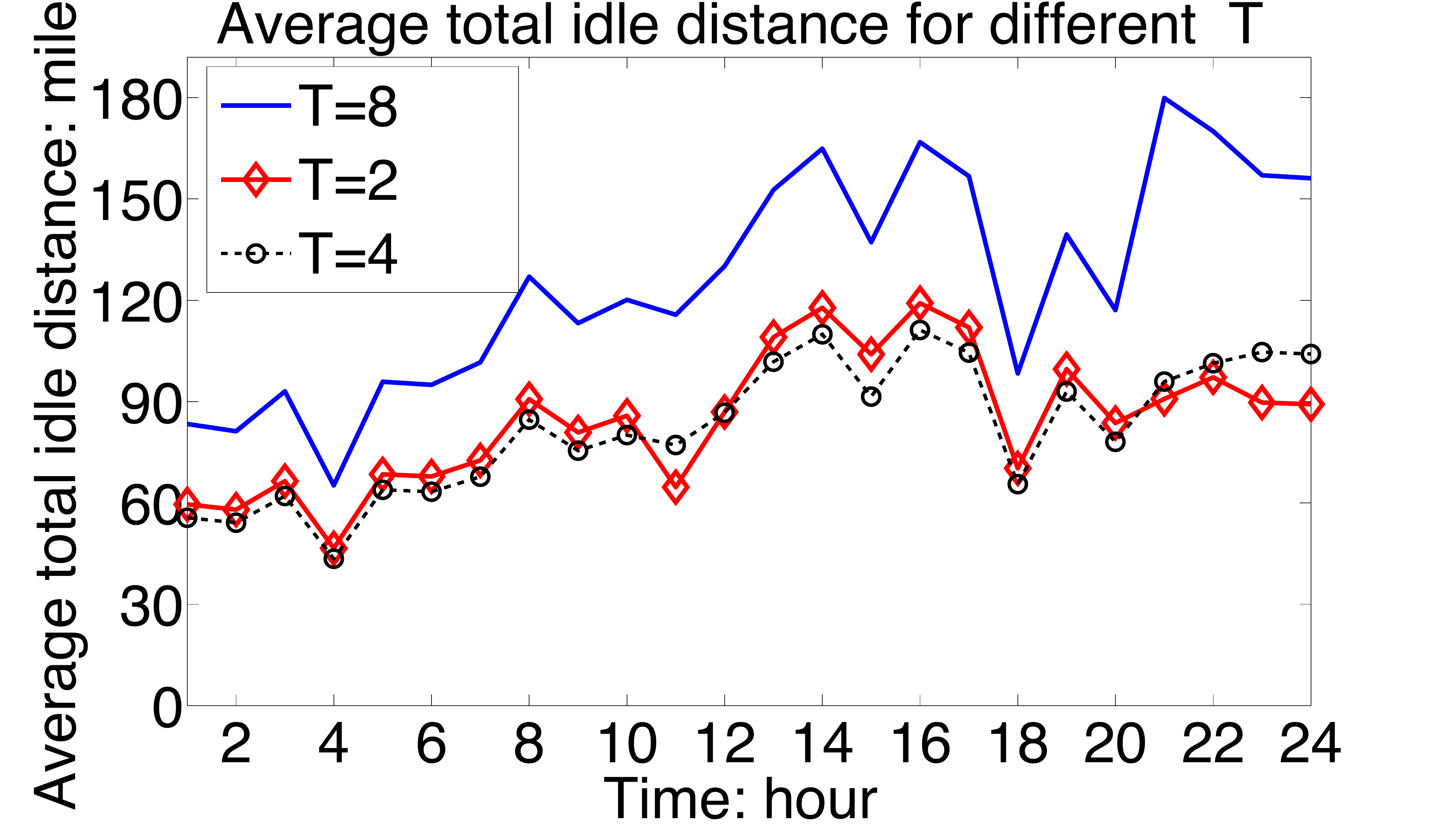}
\vspace{-10pt}
\caption{Average total idle distance at different time of one day compared for different prediction horizons.When $T=4$, idle distance is decreased at most hours compared with $T=2$. For $T=8$ the costs are worst.}
\label{compare_T}
\vspace{-15pt}
\end{figure} 
\begin{figure}[b!]
\vspace{-15pt}
\centering
\includegraphics [width=0.37\textwidth]{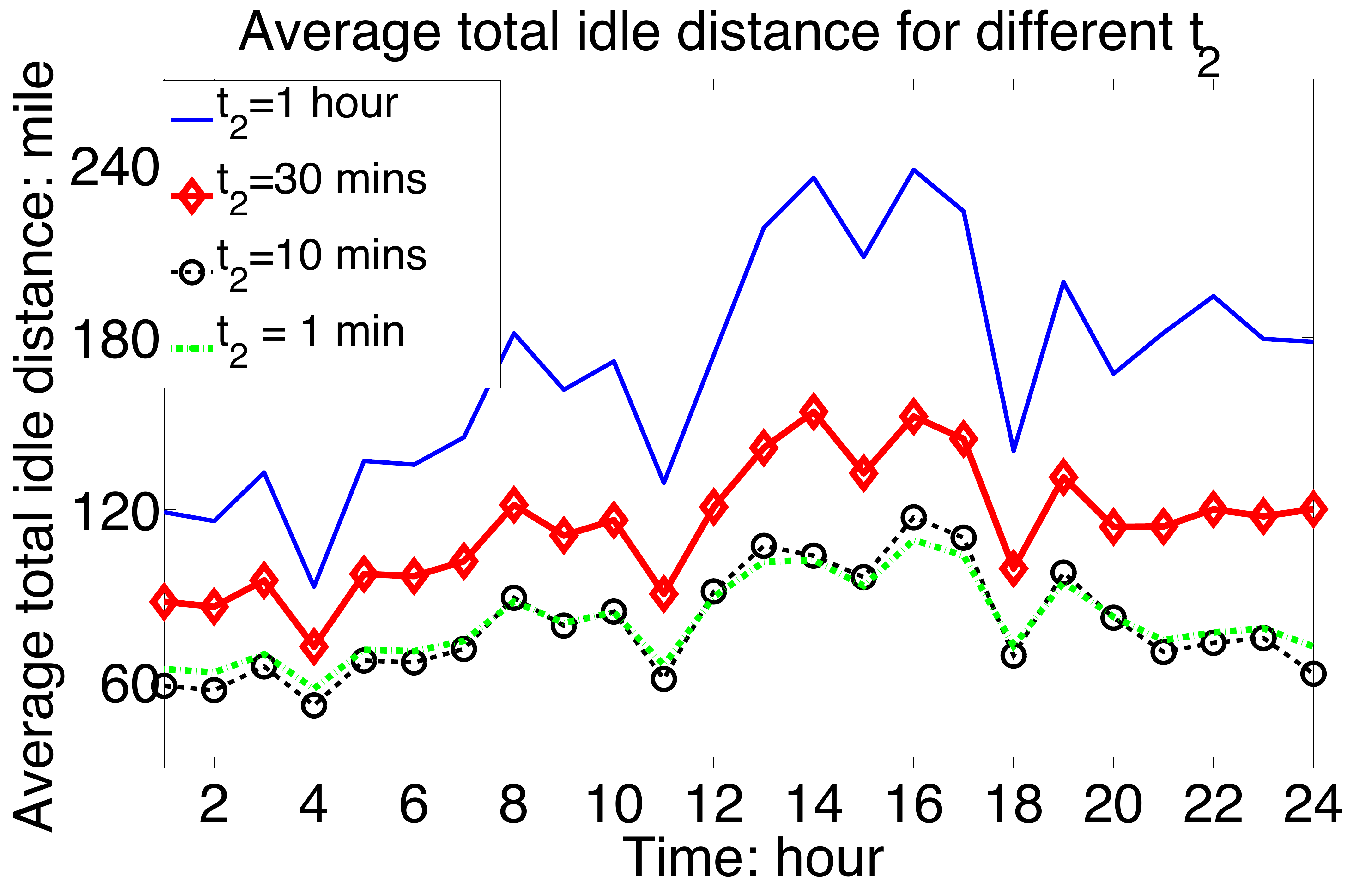}
\vspace{-10pt}
\caption{Comparison of average total idle distance for different $t_2$ -- the length of time slot for updating sensor information. With a smaller $t_2$, the cost is smaller. But when $t_2=1$ is too small to complete calculating problem~\eqref{opt2}, the dispatch result is not guaranteed to be better than $t_2=10$.}
\label{compare_t2}
%\vspace{-8pt}
\end{figure}
\begin{figure}[b!]
\vspace{-5pt}
\centering
\includegraphics [width=0.39\textwidth]{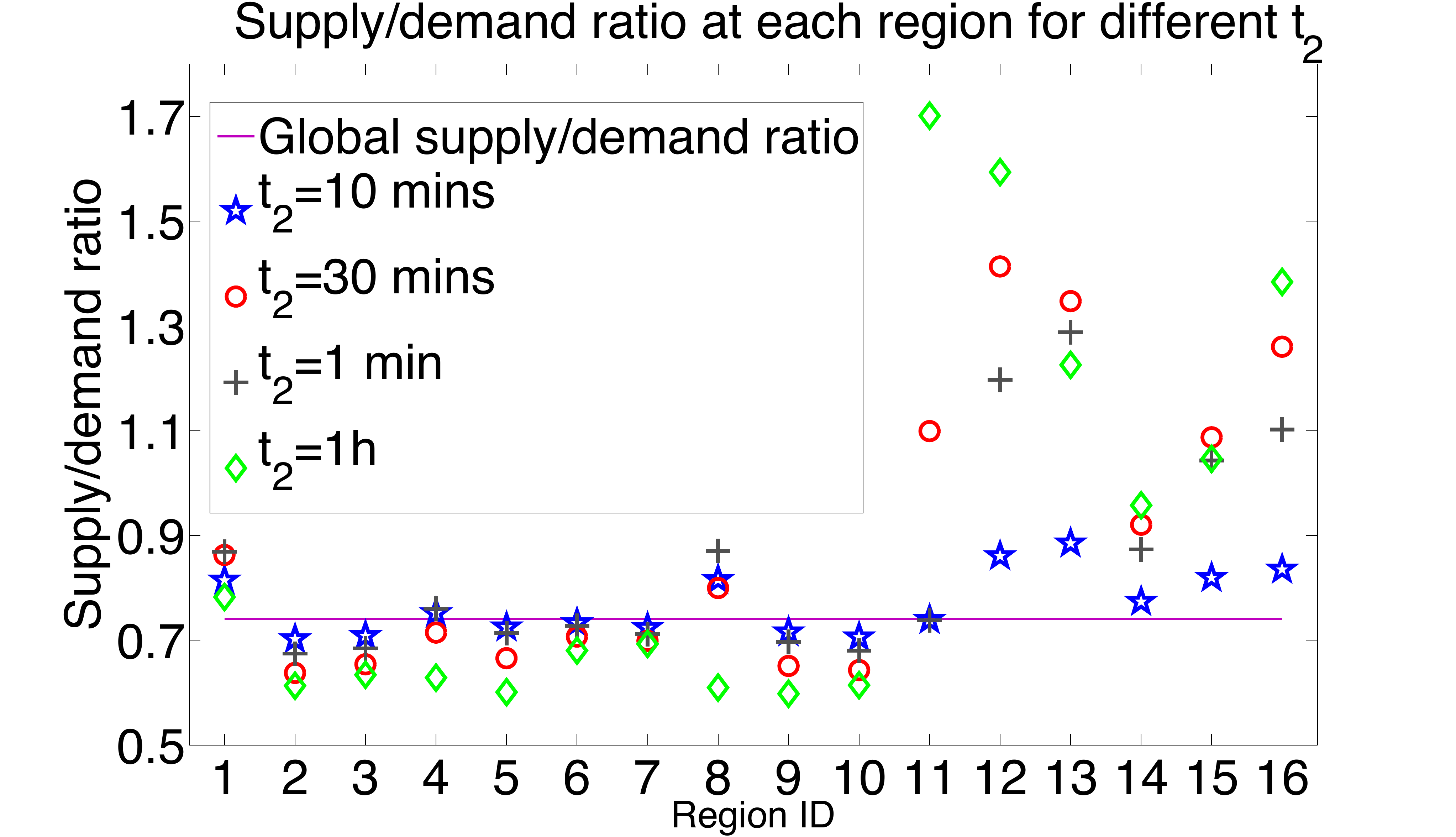}
\vspace{-10pt}
\caption{Comparison of the supply demand ratio at different regions for different lengths of time slot $t_2$. For $t_2=30$, $t_2=10$ mins and $t_2=1$ hour, results are similar. For $t_2=1$ min, the supply demand ratio is even worse at some regions, since the time slot is too short to complete one iteration.}
\label{compare_t2dist}
%\vspace{-10pt}
\end{figure}
In general, when $T$ is larger, the total idle distance to get a good supply demand ratio in future time slots should be smaller. However, when $T$ is large enough, increasing $T$ can not reduce the total idle distance any more, since the model prediction error compensates the advantage of considering future costs.
%This is because we considering possible future cost when computing current dispatch. 
For $T=2$ and $T=4$, Figure~\ref{compare_T} shows that the average total idle distance of vacant taxis at most hours of one day decreases as $T$ increases. For $T=8$ the driving distance is the largest. Theoretical reasons are discussed in Section~\ref{sec:algorithm}.

\iffalse
The main reason of reducing total idle distance with larger $T$ is that the dispatch positions also include the distance a taxi should run in the future, to  the supply/demand ratio , when the model prediction is accurate enough.
\fi

\textbf{Decide the length of time slot $t_2$}: 
For simplicity, we choose the time slot $t_1$ as one hour, to estimate requests. A smaller time slot $t_2$ for updating GPS information can reduce the total idle geographical distance with real-time taxi positions. However, one iteration of Algorithm~\ref{mpc_rt} is required to finish in less than $t_2$ time, otherwise the dispatch order will not work for the latest positions of vacant taxis, and the cost will increase. Hence $t_2$ is constrained by the problem size and computation capability. Figure~\ref{compare_t2} shows that smaller $t_2$ returns a smaller idle distance, but when $t_2=1$ Algorithm~\ref{mpc_rt} can not finish one step iteration in one minute, and the idle distance is not reduced. The supply demand ratio at each region does not vary much for $t_2=30, t_2=10$ minutes and $t_2=1$ hour, as shown in Figure~\ref{compare_t2dist}. 
Comparing two parts of costs, we get that $t_2$ mainly affects the idle driving distance cost in practice.

\section{Conclusion}
\label{sec:conclusion}
In this paper, we propose an RHC framework for the taxi dispatch problem. This method utilizes both historical and real-time GPS and occupancy data to build demand models, and applies predicted models and sensing data to decide dispatch locations for vacant taxis considering multiple objectives. From a system-level perspective, we compute suboptimal dispatch solutions for reaching a globally balanced supply demand ratio with least associated cruising distance under practical constraints. Demand model uncertainties under disruptive events are considered in the decision making process via robust dispatch problem formulations. 
By applying the RHC framework on a data set containing taxi operational records in San Francisco, we show how to regulate parameters such as objective weight, idle distance threshold, and prediction horizon in the framework design process according to experiments. Evaluation results support system level performance improvements of our RHC framework. In the future, we plan to develop privacy-preserving control framework when data of some taxis are not shared with the dispatch center.

%The MPC process uses real-time sensor data as feedback information from the taxi network to correct the prediction model.  
%The modeling process and MPC framework is 
 %and cost comparison with different given parameters.  
% In the future, we will enhance problem formulation considering information like passenger destination, road congestion, and effects of model uncertainties.  
%
%\input{appendix}
%\input{appendix
\small{
\bibliographystyle{abbrv}
\bibliography{taxi}
}

\begin{IEEEbiography}[{\includegraphics[width=1in,height=1.25in,clip,keepaspectratio]{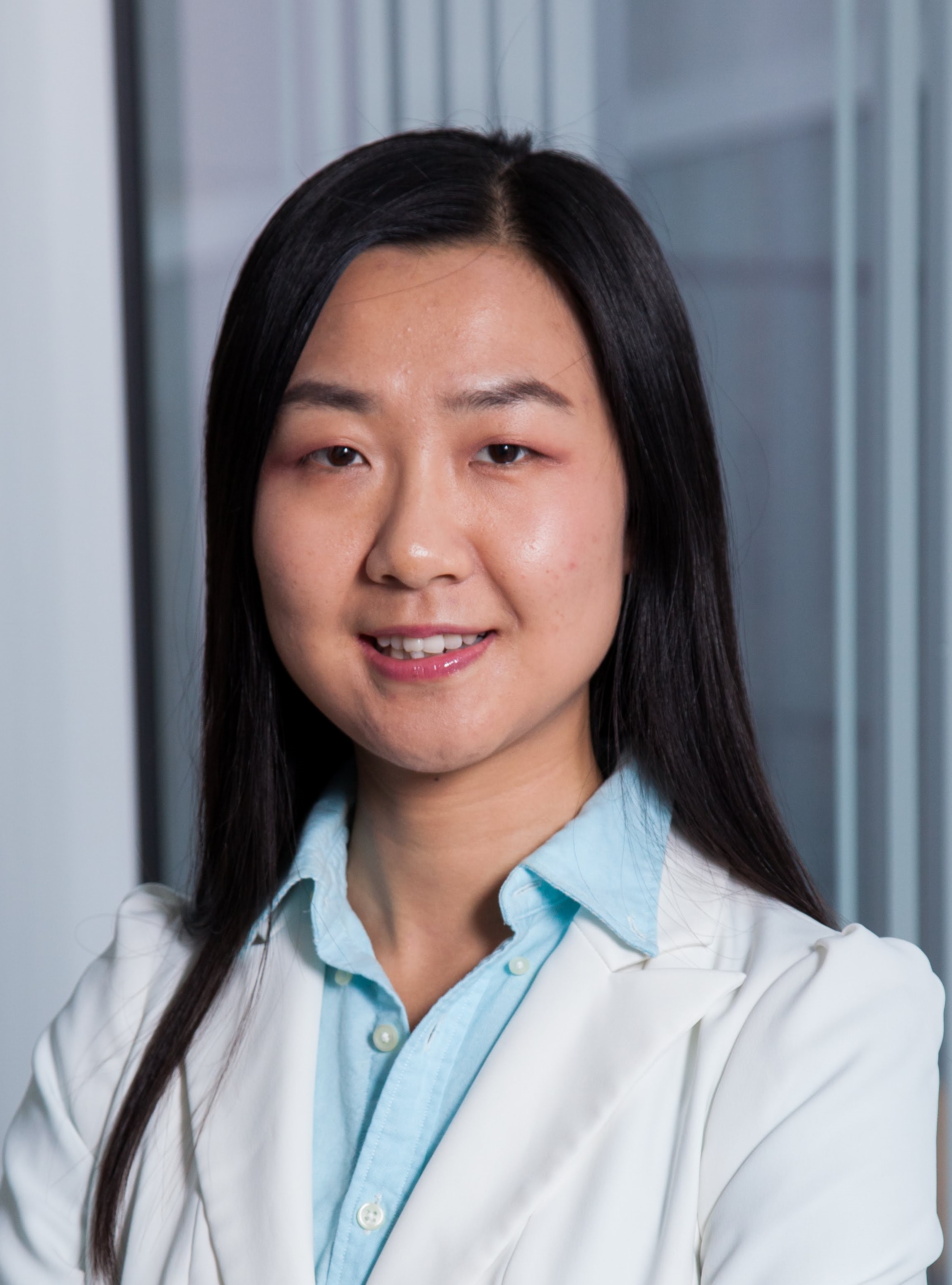}}]{Fei Miao} (S'13) received the B.Sc. degree in Automation from Shanghai Jiao Tong University, Shanghai, China in 2010. Currently, she is working toward the Ph.D. degree in the Department of Electrical and Systems Engineering at University of Pennsylvania. Her research interests include data-driven real-time control frameworks of large-scale interconnected cyber-physical systems under model uncertainties, and resilient control frameworks to address security issues of cyber-physical systems. She was a Best Paper Award Finalist at the 6th ACM/IEEE International Conference on Cyber-Physical Systems in 2015.
\end{IEEEbiography}

\begin{IEEEbiography}[{\includegraphics[width=1in,height=1.25in,clip,keepaspectratio]{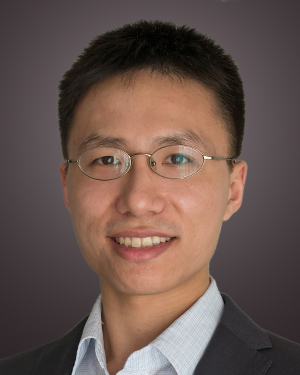}}]{Shuo Han}  (S'08-M'14) received the B.E. and M.E. degrees in Electronic Engineering from Tsinghua University, Beijing, China in 2003 and 2006, and the Ph.D. degree in Electrical Engineering from the California Institute of Technology, Pasadena, USA in 2013. He is currently a postdoctoral researcher in the Department of Electrical and Systems Engineering at the University of Pennsylvania. His current research interests include control theory, convex optimization, applied probability, and their applications in large-scale interconnected systems.
\end{IEEEbiography}

\begin{IEEEbiography}[{\includegraphics[width=1in,height=1.25in,clip,keepaspectratio]{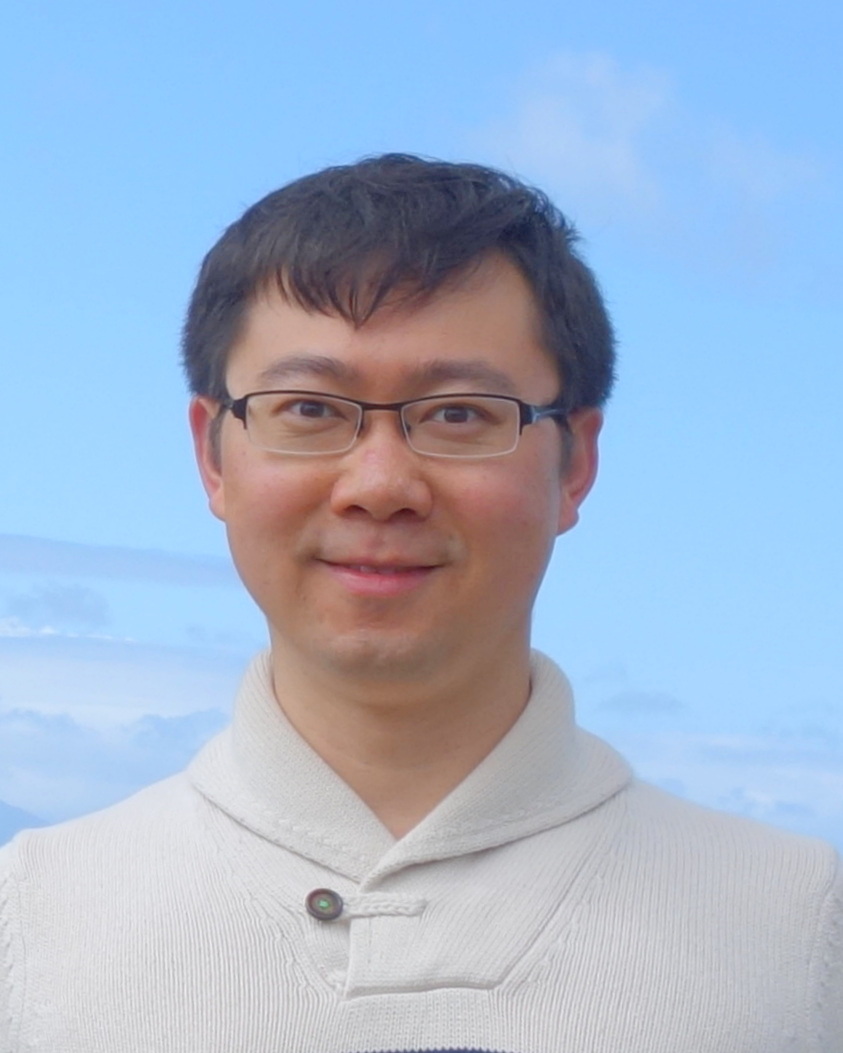}}]{Dr. Shan Lin} is an assistant professor of the Electrical and Computer Engineering Department at Stony Brook University. He received his PhD in computer science at the University of Virginia in 2010. His PhD dissertation is on Taming Networking Challenges with Feedback Control. His research is in the area of networked systems, with an emphasis on feedback control based design for cyber-physical systems and sensor systems. He works on wireless network protocols, interoperable medical devices, smart transportation systems, and intelligent sensing systems.
\end{IEEEbiography}

\begin{IEEEbiography}[{\includegraphics[width=1in,height=1.25in,clip,keepaspectratio]{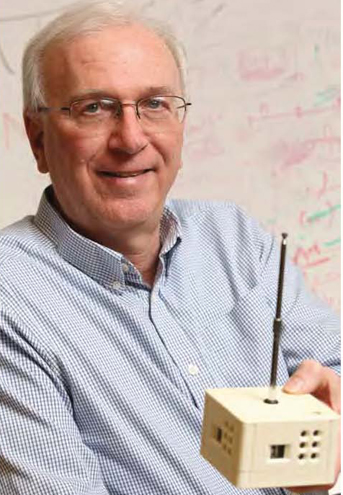}}]{Professor John A. Stankovic} is the BP America Professor in the Computer Science Department at the University of Virginia. He is a Fellow of both the IEEE and the ACM. He has been awarded an Honorary Doctorate from the University of York. He won the IEEE Real-Time Systems Technical Committee's Award for Outstanding Technical Contributions and Leadership.  He also won the IEEE Technical Committee on Distributed Processing's Distinguished Achievement Award. He has seven Best Paper awards, including one for ACM SenSys 2006. Stankovic has an h-index of 105 and over 40,000 citations. Prof. Stankovic received his PhD from Brown University.
\end{IEEEbiography}

\begin{IEEEbiography}[{\includegraphics[width=1in,height=1.25in,clip,keepaspectratio]{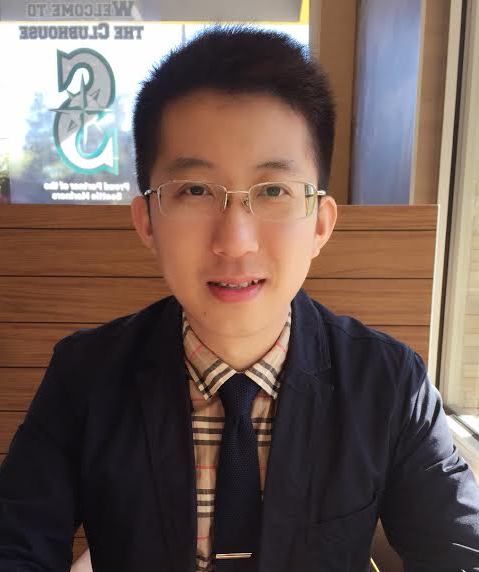}}]{Desheng Zhang} is a Research Associate at Department of Computer Science and Engineering of the University of Minnesota. Previously, he was awarded his Ph.D in Computer Science from University of Minnesota. His research includes big data analytics, mobile cyber physical systems, wireless sensor networks, and intelligent transportation systems. His research results are uniquely built upon large-scale urban data from cross-domain urban systems, including cellphone, smartcard, taxi, bus, truck, subway, bike, personal vehicle, electric vehicle, and road networks. Desheng designs and implements large-scale data-driven models and real-world services to address urban sustainability challenges.% He is a member of the IEEE.
\end{IEEEbiography}

\begin{IEEEbiography}[{\includegraphics[width=1in,height=1.1in,clip,keepaspectratio]{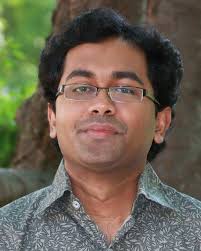}}]{Sirajum Munir} received his PhD in Computer Science from the University of Virginia in 2014. He is currently working at Bosch Research and Technology Center as a Research Scientist. His research interest lies in the areas of cyber-physical systems, wireless sensor and actuator networks, and ubiquitous computing. He has published papers in major conferences in these areas, two of which were nominated for best paper awards at ACM/IEEE ICCPS.
\end{IEEEbiography}

\begin{IEEEbiography}[{\includegraphics[width=1in,height=1.25in,clip,keepaspectratio]{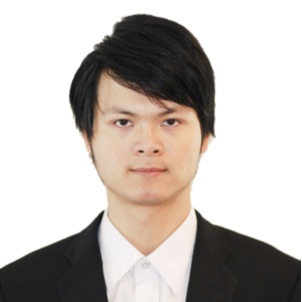}}]{Hua Huang} received the BE degree from Huazhong University of Science and Technology in 2012, MS degree in Temple University in 2014. He is currently working towards the PhD degree in the Department of Electrical and Computer Engineering in the Stony Brook University. His research interests include activity recognition in wearable devices and smart building, device-free indoor localization, deployment and scheduling in wireless sensor networks.
\end{IEEEbiography}

\begin{IEEEbiography}[{\includegraphics[width=1in,height=1.25in,clip,keepaspectratio]{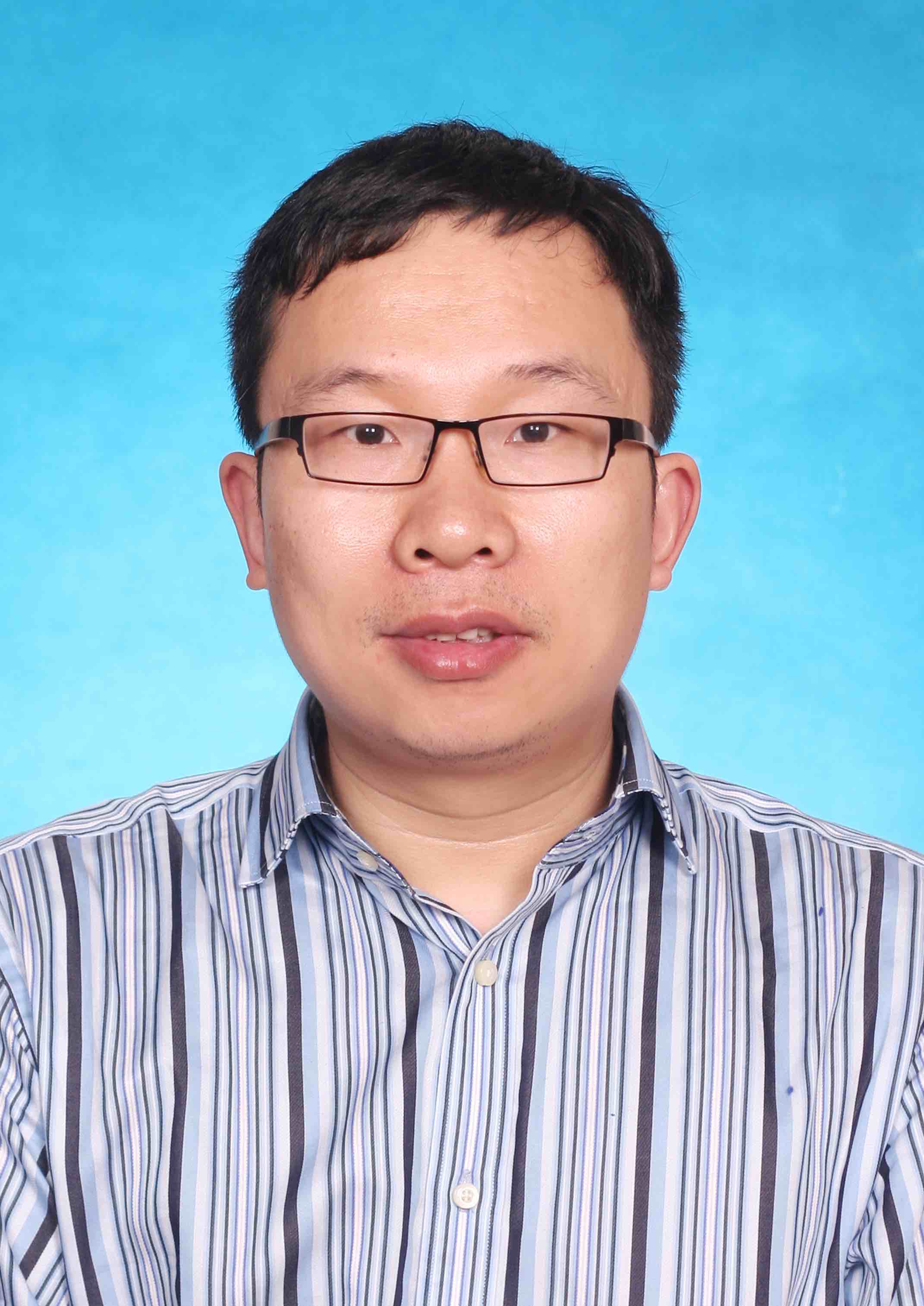}}]{Dr. Tian He} is currently an associate professor in the Department of Computer Science and Engineering at the University of Minnesota-Twin City. Dr. He is the author and co-author of over 200 papers in premier network journals and conferences with over 17,000 citations (H-Index 52). Dr. He is the recipient of the NSF CAREER Award, George W. Taylor Distinguished Research Award and McKnight Land-Grant Professorship, and many best paper awards in networking.  His research includes wireless sensor networks, cyber-physical systems, intelligent transportation systems, real-time embedded systems and distributed systems.
\end{IEEEbiography}

\begin{IEEEbiography}[{\includegraphics[width=1in,height=1.25in,clip,keepaspectratio]{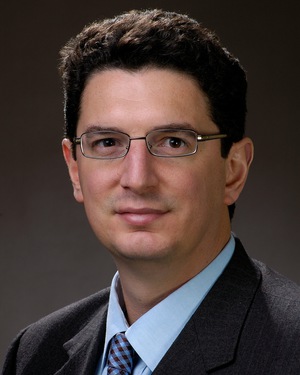}}]{George J. Pappas} (S'90-M'91-SM'04-F'09) received the Ph.D. degree in electrical engineering and computer sciences from the University of California, Berkeley, CA, USA, in 1998. He is currently the Joseph Moore Professor and Chair of the Department of Electrical and Systems Engineering, University of Pennsylvania, Philadelphia, PA, USA. He also holds a secondary appointment with the Department of Computer and Information Sciences and the Department of Mechanical Engineering and Applied Mechanics. He is a Member of the GRASP Lab and the PRECISE Center. He had previously served as the Deputy Dean for Research with the School of Engineering and Applied Science. His research interests include control theory and, in particular, hybrid systems, embedded systems, cyber-physical systems, and hierarchical and distributed control systems, with applications to unmanned aerial vehicles, distributed robotics, green buildings, and bimolecular networks. Dr. Pappas has received various awards, such as the Antonio Ruberti Young Researcher Prize, the George S. Axelby Award, the Hugo Schuck Best Paper Award, the George H. Heilmeier Award, the National Science Foundation PECASE award and numerous best student papers awards at ACC, CDC, and ICCPS.
\end{IEEEbiography}

%\balancecolumns
\end{document}